\documentclass[11pt,reqno]{amsart}
\usepackage{amsaddr}
\usepackage[left=2cm,right=2cm,top=2cm,bottom=2cm]{geometry}
\usepackage[usenames,dvipsnames]{xcolor}
\usepackage{colortbl}
\usepackage[shortlabels]{enumitem}
\usepackage{amsmath,amsthm,amssymb,bm,bbm,dsfont}
\usepackage{thmtools}
\usepackage{thm-restate}

\usepackage{mathtools}\mathtoolsset{centercolon}
\mathtoolsset{showonlyrefs}  
\usepackage[mathscr]{eucal}
\usepackage{upgreek}
\usepackage{setspace}
\usepackage{graphicx}
\usepackage{verbatim}
\usepackage{float}
\usepackage{placeins}
\usepackage{array}
\usepackage{booktabs}
\usepackage{makecell}

\usepackage{threeparttable}
\usepackage[update,prepend]{epstopdf}
\usepackage{multirow}
\usepackage{amsfonts,amssymb,dsfont}
\usepackage[abs]{overpic}

\usepackage{tikz}
\usepackage{tikzscale}
\usetikzlibrary{calc}
\usetikzlibrary{patterns}
\usetikzlibrary{decorations.pathreplacing}
\usepackage{blox}

\bXLineStyle{-}

\definecolor{dullmagenta}{rgb}{0.4,0,0.4}   
\definecolor{darkblue}{rgb}{0,0,0.4}
\usepackage{hyperref}
\hypersetup{
	pdftoolbar=true,        
	pdfmenubar=true,        
	pdffitwindow=false,     
	pdfstartview={FitH},    
	pdftitle={My title},    
	pdfauthor={Author},     
	pdfsubject={Subject},   
	pdfcreator={Creator},   
	pdfproducer={Producer}, 
	pdfkeywords={keyword1} {key2} {key3}, 
	pdfnewwindow=true,      
	linktocpage=true,
	colorlinks=true,        
	linkcolor=Red,          
	citecolor=ForestGreen,  
	filecolor=Magenta,      
	urlcolor=BlueViolet,    
}

\usepackage{doi}
\usepackage{caption, subcaption}
\usepackage{enumitem}
\usepackage{shadethm}

\makeatletter
\newcommand{\opnorm}{\@ifstar\@opnorms\@opnorm}
\newcommand{\@opnorms}[1]{%
	\left|\mkern-1.5mu\left|\mkern-1.5mu\left|
	#1
	\right|\mkern-1.5mu\right|\mkern-1.5mu\right|
}
\newcommand{\@opnorm}[2][]{%
	\mathopen{#1|\mkern-1.5mu#1|\mkern-1.5mu#1|}
	#2
	\mathclose{#1|\mkern-1.5mu#1|\mkern-1.5mu#1|}
}
\makeatother

\makeatletter
\ifx\@NODS\undefined%

\else%
\fi
\makeatother

\makeatletter
\ifx\@NODS\undefined%

\let\mathbb=\mathds
\else%
\fi
\makeatother

\usepackage{xargs}
\usepackage[prependcaption]{todonotes}
\newcommandx{\eric}[2][1=]{\todo[inline, author={Eric}, linecolor=yellow,backgroundcolor=yellow!25,bordercolor=yellow,#1]{#2}}

\newcommandx{\ericside}[2][1=]{\todo[author={Eric}, linecolor=yellow,backgroundcolor=yellow!25,bordercolor=yellow,#1]{#2}}

\DeclareMathOperator{\Tr}{Tr}

\DeclareMathOperator{\e}{\mathrm{e}}

\newcommand{\ket}[1]{| #1 \rangle}

\newcommand{\be}{{\mathbf e}}

\def\0{{\mathbf{0}}}
\def\1{{\mathbf{1}}}
\def\2{{\mathbf{2}}}
\def\3{{\mathbf{3}}}
\def\4{{\mathbf{4}}}
\def\5{{\mathbf{5}}}
\def\6{{\mathbf{6}}}

\def\7{{\mathbf{7}}}
\def\8{{\mathbf{8}}}
\def\9{{\mathbf{9}}}

\def\be{\begin{equation}}
\def\ee{\end{equation}}
\def\bea{\begin{eqnarray}}
\def\eea{\end{eqnarray}}

\def\eps{\varepsilon}

\theoremstyle{plain}
\newtheorem{theo}{Theorem} 
\newtheorem{prop}[theo]{Proposition} 
\newtheorem{lemm}[theo]{Lemma} 
\newtheorem{coro}[theo]{Corollary} 
\newshadetheorem{conj}{Conjecture}

\theoremstyle{definition}
\newtheorem{defn}[theo]{Definition} 

\theoremstyle{remark}
\newtheorem{remark}{Remark}[section]

\begin{document}
	
\let\origmaketitle\maketitle
\def\maketitle{
	\begingroup
	\def\uppercasenonmath##1{} 
	\let\MakeUppercase\relax 
	\origmaketitle
	\endgroup
}

\title{\bfseries \Large{Discrimination of quantum states \protect\\ 
                        under locality constraints in the many-copy setting}}

\author{{Hao-Chung Cheng$^{1-4}$, Andreas Winter$^{5,6}$, and Nengkun Yu$^{7,8}$}}
\address{\small  		
	$^{1}$Department of Electrical Engineering \& Graduate Institute of Communication Engineering\\
	Department of Mathematics
	\\ National Taiwan University, Tapei 106, Taiwan (R.O.C.) \\
	$^{2}$Center for Quantum Science and Engineering,  National Taiwan University, Tapei 106, Taiwan (R.O.C.)\\
	$^{3}$Physics Division, National Center for Theoretical Sciences, Taipei 10617, Taiwan (R.O.C.)\\
	$^{4}$Hon Hai (Foxconn) Quantum Computing Center, New Taipei City 236, Taiwan (R.O.C.)\\	
	$^{5}$ICREA—Instituci{\'o} Catalana de Recerca i Estudis Avan\c{c}ats, 08010 Barcelona, Spain \\
	$^{6}$F{\'i}sica Te{\`o}rica: Informaci{\`o} i Fen{\`o}mens Qu{\`a}ntics, Departament de F{\'i}sica, \\ Universitat Aut{\`o}noma de Barcelona, 08193 Barcelona, Spain \\
	$^{7}$Computer Science Department,
 Stony Brook University, NY, USA\\
	$^{8}$Centre for Quantum Software and Information \& Faculty of Engineering and Information Technology,\\ University of Technology Sydney, Ultimo NSW 2007, Australia
}

\email{\href{mailto:haochung.ch@gmail.com}{haochung.ch@gmail.com}; \href{mailto:andreas.winter@uab.cat}{andreas.winter@uab.cat}; \href{mailto:nengkunyu@gmail.com}{nengkunyu@gmail.com}}

\date{15 August 2023} 
	
\begin{abstract}


We study quantum hypothesis testing between orthogonal states under restricted local measurements in the many-copy scenario.
For testing arbitrary multipartite entangled pure state against its orthogonal complement state via the local operation and classical communication (LOCC) operation, we prove that the optimal average error probability always decays exponentially in the number of copies.
Second, we provide a sufficient condition for the LOCC operations to achieve the same performance as the positive-partial-transpose (PPT) operations. 
We further show that testing a maximally entangled state against its orthogonal complement and testing extremal Werner states both fulfill the above-mentioned condition. Hence, we determine the explicit expressions for the optimal average error probability, the optimal trade-off between the type-I and type-II errors, and the associated Chernoff, Stein, Hoeffding, and strong converse exponents.

Then, we show an infinite {asymptotic} separation between the separable (SEP) and PPT operations by providing a pair
of states constructed from an unextendible product basis (UPB). The quantum states can be distinguished perfectly by PPT operations, while the optimal error probability, with SEP operations, admits an exponential lower bound. On the technical side, we prove this result by providing a quantitative version of the well-known statement that the tensor product of UPBs is a UPB.
\end{abstract}
	
\maketitle

\section{Introduction}
Testing whether a system has a specific property is fundamental. 
In statistics, this problem is called \emph{hypothesis testing} \cite{Leh86}. It has substantial 
applications in numerous fields, such as 
information sciences \cite{Bla74, HHH07, CK11, PPV10, CH17, CHT19, Hao-Chung}, 
computational learning theory \cite{GGR98, KR98, Ron07, Ron09}, 
property and distribution testing \cite{Fis04, ADK15, Can15, Yu19}, 
and differential privacy \cite{WZ10, DJW13, KOV17, BBG+19, DRS+19}.

The most basic form of hypothesis testing is binary hypothesis testing, i.e.,~ testing a null hypothesis 
$\mathsf{H}_0$ against an alternative hypothesis $\mathsf{H}_1$.
In quantum computing, the two hypotheses are modeled by quantum states $\rho_0$ and $\rho_1$, 
respectively. 
To distinguish the two quantum states, one has to perform a test $T$, or equivalently, a two-outcome 
positive-operator valued measure (POVM) measurement $\{T, \mathds{1}- T \}$ on the received 
state, where $0\leq T \leq \mathds{1}$ is a quantum observable. 
Such a test $T$ incurs two types of errors: the type-I error probability 
$\alpha(T) := \Pr\left[\mathsf{H}_1|\mathsf{H}_0 \right] = \Tr\left[ \rho_0 (\mathds{1} - T) \right]$, the probability of accepting $\mathsf{H}_1$ when $\mathsf{H}_0$ is true; and the type-II error probability
$\beta(T) := \Pr\left[\mathsf{H}_0|\mathsf{H}_1 \right] = \Tr\left[ \rho_1 T \right]$, the probability of deciding
for $\mathsf{H}_0$ when actually $\mathsf{H}_1$ is true.
Subsequently, we simply term $\alpha(T)$ and $\beta(T)$ as the type-I error and the type-II error, respectively.
If the prior probabilities of the hypotheses are known, say $p$ and $1-p$, we measure the performance 
of the decision scheme by calculating the average (Bayes) error probability. We name this the Bayesian 
approach, and specifically the symmetric setting, $p=1-p=\frac12$, when the two hypotheses are equally likely.
In most practical situations where the prior probabilities are unknown, the Neyman--Pearson approach is used to analyze the trade-off between the two types of errors. We name this the asymmetric setting.

In the \emph{Bayesian} setting, the optimal average error probability is defined as
\[
P_\text{e}(\rho_0, \rho_1; p ) := \inf_{0\leq T\leq \mathds{1}} \left\{  p \alpha (T) + (1-p) \beta (T) \right\}
\]
for $ p:= \Pr[\mathsf{H}_0] \in (0,1)$.
Helstrom and Holevo \cite{Hel76} proved a closed-form expression of $P_\text{e}$ and showed that the
optimal test is achieved by projection onto the positive support of $p \rho_0 - (1-p) \rho_1$. 
This measurement can be viewed as a quantum generalization of the classical 
\emph{Neyman--Pearson test}, as described in \cite{NP33, Che52, Che56, Hoe65}.

In the present paper, we are concerned with the many-copy and asymptotic behavior of 
$P_\text{e}(\rho_0^{\otimes n}, \rho_1^{\otimes n}; p)$, where $n$ identical copies of states are prepared.
Then, the celebrated quantum Chernoff theorem \cite{ACM+07, ANS+08, NS09} establishes that
\begin{align}
\mathrm{Chernoff}(\rho_0, \rho_1) &:= \lim_{n\to \infty} - \frac1n \log P_\text{e} (\rho_0^{\otimes n}, \rho_1^{\otimes n}; p ) = - \min_{0\leq s \leq 1} \log \Tr\left[ \rho_0^{1-s} \rho_1^{s}\right]. \label{eq:Chernoff}
\end{align}
That is, the \emph{Chernoff exponent} determines the convergence rate of 
the error probability. 
{Because of this result in the asymptotic setting, we consider all the Bayesian error 
probabilities $P_\text{e}(\rho_0^{\otimes n}, \rho_1^{\otimes n}; p)$ as pertaining to
the symmetric setting, as long as $p$ does not depend on $n$.}

In the \emph{asymmetric} setting, one aims to study the asymptotic functional dependence 
between the type-I and type-II errors.
In particular, three exponents are the most important. Detailed definitions of exponents
will be given in Section~\ref{sec:notation}.
The \emph{Stein exponent} characterizes the best (i.e., ~the largest) exponential decay
rate of the type-II error when the type-I error upper is bounded by some $\eps \in(0,1)$.
Quantum Stein's lemma \cite{HP91,ON00} shows that the exponent is given by the quantum 
relative entropy \cite{Ume62}. Moreover, it is independent of $\eps \in (0,1)$, which is known as 
the strong converse property.
The \emph{Hoeffding exponent} and the \emph{strong converse exponent}, respectively, determine 
the optimal exponential rate of the type-I error (or the type-I success probability) when the 
type-II error exponentially decays at the rate below or above the Stein exponent.
They are proved to be given by quantities involving Petz's R\'enyi divergence \cite{ANS+08,NS09, Pet86, Nag06} 
and the sandwiched R\'enyi divergence \cite{Nag05, MO14, MDS+13, WWY14}.
Other extensions in the large, moderate, and small deviation regimes have been studied 
in depth by~\cite{Hay02, BS03, BDJ+05, NH07, TH13, Li14, CH17, CTT2017, CHT19, Hao-Chung}.

Small to intermediate scale quantum computers will be available in the near-term future \cite{BIS+18, Pre18}. 
However, such quantum computers will be built in geographically separated laboratories, which 
means that each lab will perform only local quantum operations, and mutual classical communication may be
available. 
These operations constitute a restricted class of measurements known as \emph{local 
operations and classical communication} (LOCC) \cite{BBC+93, NC09, HHH+09, CLM+14}. Naturally, this leads to the question of how well do LOCC measurements perform in hypothesis testing compared to global measurements? For instance, the above-mentioned quantum Neyman--Pearson test is generally not implementable 
via LOCC due to quantum entanglement and nonlocality \cite{BDF+99, HHH+09}.
Therefore, the problem of local discrimination and local hypothesis testing recently gained considerable attention in the field of quantum computation and quantum information \cite{WM08, BHL+14, OH14, OH15, HO17, CMM+08, CVM+10, Nat10, YDY12, YDY14, LWD17, YZ17, AK18, CDR19a, CDR19b}.
Unfortunately, only limited results are known due to the complicated mathematical structure of LOCC.

In the present work, our goal is to study the asymptotic behavior of the errors incurred by 
the restriction to LOCC and derive the above-mentioned four exponents.
We also consider other classes of measurements: the positive-partial-transpose operations (PPT) 
and separable operations (SEP) \cite{Rai01}, mainly out of theoretical interest or as a tool to analyze LOCC. It is well-known that strict inclusions hold among them \cite{BDF+99}, i.e.
\begin{align} \label{eq:relation}
	\text{LOCC} \subset \text{SEP} \subset \text{PPT} \subset \text{ALL}. 
\end{align}

Although hypothesis testing under LOCC has been studied in many papers in the one-shot setting
\cite{WM08, BHL+14, OH14, OH15, HO17, CMM+08, CVM+10, Nat10, YDY14, LWD17, AK18}, the explorations of the asymptotic error behavior are relatively limited \cite{WM08, OH14, HO17}.
In this work, we focus on demonstrating an interesting phenomenon in the many-copy setting. Namely, the results show that distinguishing a pair of orthogonal states under LOCC indeed exhibits a fundamental 
difference from the conventional task using global measurements.
By definition of the Chernoff exponent given in Eq.~\eqref{eq:Chernoff} and the 
inclusions Eq.~\eqref{eq:relation}, we have 
\begin{align} \label{eq:ordering}
\mathrm{Chernoff}^\text{LOCC} \leq \mathrm{Chernoff}^\text{SEP} \leq \mathrm{Chernoff}^\text{PPT}  \leq  \mathrm{Chernoff}^\text{ALL} 
= - \min_{0\leq s \leq 1} \log \Tr\left[ \rho_0^{1-s} \rho_1^{s}\right].
\end{align}
Here and subsequently, we put a superscript `X' on $P_\text{e}$ and the Chernoff exponent to highlight the class X measurements allowed. Particularly intriguing examples arise in the context of data 
hiding \cite{TDL01, DLT02, WM08, MWW09, LW13, LPW18} where the underlying states are orthogonal, 
i.e., $\Tr\left[ \rho_0 \rho_1 \right] = 0$.
This result implies that $\mathrm{Chernoff}^\text{ALL}(\rho_0,\rho_1) = \infty$.
However, for a given 
pair of orthogonal states with entanglement, the problem of whether $\mathrm{Chernoff}^\text{LOCC}(\rho_0,\rho_1)$ is finite or not  has remained open since the early days of quantum information 
theory. We know that $\mathrm{Chernoff}^\text{LOCC}(\rho_0,\rho_1)$ is nonzero, which can be achieved by implementing local tomography, then applying the classical Chernoff bound. 
On the other hand, it is unclear if the definition of the Chernoff exponent for a class $\text{X}\in\{\text{LOCC},\text{SEP},\text{PPT}\}$ is faithful, in the sense that the 
sequence $\left( - \frac1n \log P_\text{e}^\text{X} (\rho_0^{\otimes n}, \rho_1^{\otimes n}; p ) \right)_n$ 
diverges if and only if $P_\text{e}^\text{X} (\rho_0^{\otimes n}, \rho_1^{\otimes n}; p ) = 0$
for some finite $n$; cf.~\cite{YZ17}\footnote{In \cite{YZ17}, the authors showed that if two quantum operation cannot be perfectly distinguishable with finite uses, then the associated Chernoff exponent is finite.}.
Hence, in this work, we study the case of distinguishing an entangled state (possibly on a 
multipartite system) and its orthogonal complement using restricted classes of POVMs.
In particular, we show that the Chernoff bounds in this case is faithful for all the three classes;
we remark that no simple expression for a general multipartite entangled state is known. 

Although it has been shown that the inclusion relations Eq.~\eqref{eq:relation} are all 
strict \cite{BDF+99}, it is not known whether strict inclusion still holds in the many-copy 
asymptotics. For instance, does any equality hold in Eq.~\eqref{eq:ordering}? This question naturally
arises since SEP and PPT operations are often exploited to approximate LOCC operations. 
Hence, one may ask how differently the restricted classes of measurements perform. 

	In this paper, we further show that the (one-way) LOCC operations could achieve the same performance
as the PPT operation, even if the state is highly entangled. In this case, the equalities in
(2) between LOCC, SEP, and PPT hold. On the other hand, we demonstrate that there is an infinite separation between the SEP and PPT operations, even if there is no entanglement
involved in the underlying states.
Namely, we 
construct a pair of states such that $\mathrm{Chernoff}^\text{PPT}  =\infty$, while 
$\mathrm{Chernoff}^\text{SEP} \leq -\log \mu < \infty$ for some $\mu>0$.

\medskip
\textit{Organization of the paper.}
We summarize our main contributions in Section~\ref{sec:contributions}.
In Section~\ref{sec:comparisons}, we compare our results with existing works.
Section~\ref{sec:notation} introduces necessary notation and definitions.
In Section~\ref{sec:general}, we consider hypothesis testing between arbitrary entangled pure state and its orthogonal complement. 
In Section~\ref{sec:optimal_LOCC}, we study the scenario when LOCC equals PPT POVMs.
A special case of testing a maximally entangled state and its orthogonal complement is shown in Section~\ref{sec:MES}, and testing extremal Werner states is shown in Section~\ref{sec:Werner}. 
In Section~\ref{sec:separation}, we demonstrate an infinite {asymptotic} separation for the SEP and PPT operations. In closing, we provide discussions and conclusions in Section~\ref{sec:conclusions}.
Appendix~\ref{sec:MES_high} includes the case of testing a pure state with equal non-zero Schmidt coefficients.

\subsection{Main Contributions} \label{sec:contributions}

This work concerns hypothesis testing between orthogonal states using restricted measurements and the relations between LOCC, SEP, and PPT in the many-copy scenario.
Below we list our contributions and findings regarding to these questions.

\begin{enumerate}[(I)]
	\item 
	We study the optimal error probability of distinguishing an arbitrary multipartite entangled pure state $\psi$ and its orthogonal complement $\psi^{\perp} := \tfrac{\mathds{1}- \psi}{D -1 }$ (where $D$ is the dimension of the underlying Hilbert space) in the many-copy scenario.
\begin{enumerate}[(i)]
	\item 
	We prove that the optimal error probability always decays exponentially in the number of copies $n$ (Theorem~\ref{theo:faithfulness_Chernoff}, Section~\ref{sec:general}): There exists constants $0<a\leq b < \infty$ such that,
	\begin{align} \label{eq:main}
		\e^{-nb} \leq
		 	P_\textnormal{e}^\textnormal{LOCC}\left( \psi^{\otimes n} , (\psi^\perp)^{\otimes n}; p \right) \leq (1-p) \cdot \e^{-n a}, \quad \forall n\in\mathbb{N}.
	\end{align}
	Hence, the Chernoff bound is faithful.
	We remark that the derived error exponent $a$ in the upper bound depends only on the dimension of the underlying Hilbert spaces, and it is tight if $\psi$ is a maximally entangled state.
	
	Our key technique proving this result is by establishing an exponential lower bound to the optimal error probability for distinguishing a bipartite entangled state $\psi$ on a bipartite system $\mathbb{C}^d\otimes \mathbb{C}^d$ against its orthogonal complement $\psi^{\perp}$ using PPT POVMs (Proposition~\ref{prop:Exp_LB}, Section~\ref{sec:converse_sym}) expressed as:
	\begin{align}
		P_\textnormal{e}^\textnormal{PPT}\left( \psi^{\otimes n} , (\psi^\perp)^{\otimes n}; p \right)  \geq		\min\{1-p,p\} \cdot t^n, \quad \forall n\in\mathbb{N}
	\end{align}
where $t := \tfrac{1-\eta}{(d^2-1) \eta } \in (0,\tfrac{1}{d+1}]$ and $\eta$ denotes the largest squared Schmidt coefficient of $\psi$.
	
	The above bound also provides a lower bound to the error probability using LOCC POVMs.

	Moreover, our approach extends to a strong converse bound in the asymmetric setting. Specifically, the type-I error approaches $1$ at the exponential rate  $r - \log \tfrac1t$ whenever the type-II error decays exponentially at rate $r > \log \tfrac1t$, which also implies that the Stein exponent is upper bounded by the quantity $\log \tfrac1t$ (i.e., the strong converse property).

\item 

We find a sufficient condition to characterize when the three classes of measurement, LOCC, SEP, and PPT, collapse in either symmetric or asymmetric hypothesis testing (Theorem~\ref{theo:sufficient_condition}, Section~\ref{sec:optimal_LOCC}).
Moreover, we show that such a sufficient condition is fulfilled 
when the underlying pair of states are maximally entangled state
$\Phi_d := \frac1d \sum_{i,j=0}^{d-1} |ii\rangle\langle jj| $ on $\mathbb{C}^d\times \mathbb{C}^d$ and its orthogonal complement state $\Phi_d^\perp:= (\mathds{1} - \Phi_d)/(d^2-1)$.

We then explicitly calculate the average error probability (Theorem~\ref{theo:MES}, Section~\ref{sec:MES}) under an optimal LOCC protocol:
\begin{align}
P_\textnormal{e}^\textnormal{LOCC} (\Phi_d^{\otimes n}, (\Phi_d^\perp)^{\otimes n}; p ) = \min\left\{ (1-p) \left( \frac{1}{d+1} \right)^n, p 
\right\}.
\end{align}
This then immediately gives the Chernoff exponent: $\mathrm{Chernoff}^\text{LOCC} (\Phi_d, \Phi_d^\perp) = \log (d+1)$.

Further, we establish the optimal trade-off between the type-I and type-II errors, i.e.,~given any type-I error no larger than a constant $\varepsilon $, the minimum type-II error of using a restricted class `X' of measurements, denoted by $\beta_{\varepsilon}^{\mathrm{X}}(\Phi_d^{\otimes n}, (\Phi_d^{\perp})^{\otimes n} )$, satisfies
\begin{align} \label{eq:trade-off0}
\beta_{\varepsilon}^{\mathrm{X}}(\Phi_d^{\otimes n}, (\Phi_d^{\perp})^{\otimes n} ) = \frac{1-\varepsilon}{(d+1)^n}.
\end{align}
With such functional dependence and exact characterization, we then obtain the associated Stein, Hoeffding, and strong converse exponents in Corollaries~\ref{coro:MES_Stein}, \ref{coro:MES_Hoeffding}, and \ref{coro:MES_sc}, respectively.
We remark that, in the general cases of testing arbitrary states using restricted measurements, the single-letter formula is unknown even for the Stein exponent. Moreover, the known Stein exponent for the general case only holds at $\varepsilon = 0$, which asserts that the type-I error vanishes asymptotically as the type-II error decays at a rate below the Stein exponent.
Nevertheless, it does not provide the strong converse property; namely, the minimum type-I error converges to $1$ as the type-II error decays at a rate above the Stein exponent.
We refer the reader to Section~\ref{sec:comparisons} for more detailed discussions.

\item Our results apply to the case of testing the pure state with uniform nonzero Schmidt coefficients (i.e.~$\frac1m \sum_{i,j=0}^{m-1} |ii\rangle\langle jj|$ for $m\leq d$) and its orthogonal complement (Propositions~\ref{theo:MES_high} and \ref{theo:MES_high_1} of Section~\ref{sec:MES_high}).
Lastly, we extend Matthews and Winter's work \cite{WM08} to show that testing extremal Werner states \cite{Wer89} 
also satisfies the above sufficient condition.
Therefore, we establish the optimal trade-off between the type-I and type-II errors and the corresponding Stein, Hoeffding, and strong converse exponents (Theorem~\ref{theo:Werner}, Section~\ref{sec:Werner}).
In Table~\ref{table:summary} below, we summarize the established exact characterizations for the exponents in various testing setting.

\end{enumerate}

\item 
In Section~\ref{sec:separation}, we establish an \emph{infinite asymptotic separation} 
between SEP measurements and PPT measurements. Specifically, we consider the null hypothesis to be the 
uniform mixture of an unextendible product basis (UPB) \cite{BDM+99} and the alternative hypothesis to be 
a state supported on the orthogonal complement of the former. 
Such a pair of states can be discriminated perfectly by a PPT measurement, while we show that the optimal 
error probability under SEP measurements possesses an exponential lower bound (Theorem~\ref{theo:separation} 
of Section~\ref{sec:separation}).
Our key technique to establishing this result is introducing a novel quantity, \emph{unextendibility},
(Definition~\ref{defn:delta}, Section~\ref{sec:separation}) to measure how far a UPB
is from being an extendible product basis and proving its multiplicativity property under the tensor product (Proposition~\ref{prop:multiplicative}). Accordingly, our result gives a quantitative characterization of when the tensor product of UPBs is still a UPB \cite{DMS+03}.

\end{enumerate}

\begin{table}[th!]
	\centering
	\resizebox{1\columnwidth}{!}{
		\begin{tabular}{@{}c|c|c|c|c@{}} %
			\toprule
			
			Setting \textbackslash\,  Exponents &  Chernoff   &   Stein  &  Hoeffding  &  Strong converse  \\
			
			\midrule
			\midrule		
			
			$ \begin{cases} \mathsf{H}_0 :   \Phi_d^{\otimes n} \\ \mathsf{H}_1: \left( \Phi_d^\perp \right)^{\otimes n} \\ \end{cases} $ & \multirow{4}{*}{$\displaystyle \log (d+1)$} &  $\displaystyle\log(d+1)$ &  $\displaystyle \infty$ &  $\displaystyle r - \log (d+1)$ \\
			\cmidrule{1-1} \cmidrule{3-5}			
			
			$ \begin{cases} \mathsf{H}_0 :  \left( \Phi_d^\perp \right)^{\otimes n}   \\ \mathsf{H}_1: \Phi_d^{\otimes n} \\ \end{cases} $  &  &  $\displaystyle\infty$ &   $ \displaystyle\log (d+1)$ &  $\displaystyle 0$  \\
			\midrule	
			
			$ \begin{cases} \mathsf{H}_0 :   \sigma_{\text{s},d}^{\otimes n} \\ \mathsf{H}_1:  \sigma_{\text{a},d} ^{\otimes n} \\ \end{cases} $ & \multirow{4}{*}{$\displaystyle \log  \frac{d+1}{d-1} $ \cite{WM08} }  & $ \displaystyle\infty$ &   $ \displaystyle\log  \frac{d+1}{d-1} $ &  $\displaystyle 0$ \\
			\cmidrule{1-1} \cmidrule{3-5}					
			
			$ \begin{cases} \mathsf{H}_0 : \sigma_{\text{a},d}^{\otimes n}  \\ \mathsf{H}_1: \sigma_{\text{s},d}^{\otimes n} \\ \end{cases} $ &  &  $\displaystyle\log \frac{d+1}{d-1} $ &  $\displaystyle  \infty$ &  $\displaystyle r - \log  \frac{d+1}{d-1} $ \\
			\midrule
			
			$ \begin{cases} \mathsf{H}_0 :  \left( \Phi_m \oplus \mathbb{O} \right)^{\otimes n} \\ \mathsf{H}_1: \left( \lambda\Phi_m^\perp \oplus (1-\lambda)\tau \right)^{\otimes n} \\ \end{cases} $  &    $\displaystyle\log \frac{m+1}{\lambda}$ &   $\displaystyle\log  \frac{m+1}{\lambda} $ &  $\displaystyle\infty$ &   $\displaystyle r - \log \frac{m+1}{\lambda}$\\
			\midrule
			
			$ \begin{cases} \mathsf{H}_0 :  \left( \Phi_m^\perp \oplus \mathbb{O} \right)^{\otimes n} \\ \mathsf{H}_1: \left( \lambda\Phi_m \oplus (1-\lambda)\tau \right)^{\otimes n} \\ \end{cases} $  &    $\displaystyle  \max\left\{\log(m+1), \log\frac{1}{\lambda} \right\}$ &   $\infty$ &  $ \displaystyle \begin{cases} \infty & r \leq \log \frac{1}{\lambda} \\  \log (m+1) & r>\log \frac{1}{\lambda} \\ \end{cases} $ &   $\displaystyle 0$  \\
			\midrule
			
			$ \begin{cases} \mathsf{H}_0 :  \left( \sigma_{\text{s},m} \oplus \mathbb{O} \right)^{\otimes n} \\ \mathsf{H}_1: \left( \lambda\sigma_{\text{a},m} \oplus (1-\lambda)\tau \right)^{\otimes n} \\ \end{cases} $ &    $ \displaystyle \max\left\{\log \frac{m+1}{m-1} , \log\frac{1}{\lambda} \right\}$ &   $\displaystyle  \infty$ &  $\displaystyle  \begin{cases} \infty & r \leq \log \frac{1}{\lambda} \\  \log \frac{m+1}{m-1} & r>\log \frac{1}{\lambda} \\ \end{cases} $ &   $\displaystyle 0$  \\
			\midrule
			
			$ \begin{cases} \mathsf{H}_0 :  \left( \sigma_{\text{a},m} \oplus \mathbb{O} \right)^{\otimes n} \\ \mathsf{H}_1: \left( \lambda\sigma_{\text{s},d} \oplus (1-\lambda)\tau \right)^{\otimes n} \\ \end{cases} $ &  $\displaystyle \log \frac{m+1}{\lambda(m-1)} $ & $\displaystyle\log \frac{m+1}{\lambda(m-1)} $ & $ \displaystyle \infty $ & $\displaystyle r - \log \frac{m+1}{\lambda(m-1)}$  \\
			
			\bottomrule				
		\end{tabular}
	} 
	\caption{Exact characterizations of the Chernoff exponent, Eq.~\eqref{eq:Chernoff_X}, Stein exponent, Eq.~\eqref{eq:Stein}, Hoeffding exponent, Eq.~\eqref{eq:Hoeffding}, and strong converse exponent, Eq.~\eqref{eq:sc} under various settings of binary hypothesis testing via LOCC, SEP, and PPT measurements. 
	Namely, the established exponents in the table are all the same for the three classes of the restricted measurements.
	The states $\Phi_d$, $\Phi_d^\perp$, $\sigma_{\text{s},d}$, and $\sigma_{\text{a},d}$ are the 
	maximally entangled state, its orthogonal complement state, 
	the completely symmetric Werner state
	and the completely anti-symmetric Werner state \cite{Wer89}
	on $\mathbb{C}^d \otimes \mathbb{C}^d$, respectively;
	and $\tau$ is the completely mixed state of another system (whose dimension is irrelevant here). 
	The parameter $\lambda\in[0,1]$ is an arbitrary scalar.
	}	\label{table:summary}
	
\end{table}

\subsection{Comparisons to existing results} \label{sec:comparisons}

The Stein exponent for binary quantum hypothesis testing when $\eps = 0$, i.e.~$\mathrm{Stein}^\text{X} (\rho_0, \rho_1, 0)$,  was studied by Brand{\~a}o \emph{et al.} \cite{BHL+14} and proved to be given by the regularized relative entropy between the measurement outcomes, i.e.~
\begin{align}
	\lim_{\eps \to 0} \mathrm{Stein}^\text{X} (\rho_0, \rho_1,\eps) = \lim_{n\to\infty} \sup_{ \mathcal{M} \in \text{X} } \frac{ D(\mathcal{M}(\rho_0^{\otimes n})\| \mathcal{M}(\rho_1^{\otimes n}) )}{n},
\end{align}
where $D$ is the quantum relative entropy \cite{Ume62}, and $\mathcal{M}$ is any POVM in the class `$\text{X}$', and the precise definition for $\mathrm{Stein}^\text{X} (\rho_0, \rho_1, \eps)$ will be introduced later in Section~\ref{sec:notation}.
First, calculating such a regularized quantity is computationally intractable. 
Second, it is not known whether the strong converse property holds. Namely, if $\mathrm{Stein}^\text{X} (\rho_0, \rho_1,\eps)$ is dependent on $\eps \in (0,1)$.
Hence, in Corollaries~\ref{coro:MES_Stein} and \ref{coro:Werner_Stein}, we establish a single-letter formula for testing $\Phi_d$ against $\Phi_d^\perp$, and the completely anti-symmetric state against the completely symmetric state, respectively.
Moreover, in Corollaries~\ref{coro:MES_sc} and \ref{coro:Werner_sc} establish the corresponding strong converse exponents, which characterize how fast the type-I errors approach one whenever the type-II error decays too fast.

Owari and Hayashi studied the Chernoff, Stein, and Hoeffding exponents for binary hypothesis testing between an arbitrary bipartite pure state and the white noise state (i.e.,~the maximally mixed state), under one-way LOCC, two-way LOCC, and SEP POVMs \cite{OH14, OH15, HO17}. 
The authors considered the one-copy setting in Ref.~\cite{OH15}. Specifically, the authors showed that the hypothesis testing using one-way LOCC POVMs is equivalent to a \emph{classical} hypothesis testing between a probability distribution defined by the Schmidt coefficients of the bipartite pure state and the classical white noise (i.e.,~ the uniform distribution).
On the other hand, hypothesis testing using SEP POVMs is equivalent to a hypothesis-testing problem with a composite null hypothesis under global POVMs, in which solving the latter problem is easier than solving the original one.
Then, Ref.~\cite{OH14} extends the analysis of \cite{OH15} to the many-copy scenario to obtain the Chernoff, Stein, and Hoeffding exponents for several restricted POVMs. For one-way LOCC POVMs, the hypothesis-testing problem is essentially classical, while for SEP POVMs, additional large-deviation-type techniques were used. In particular, the authors proved that the Stein exponent is the same for all three classes of POVMs.
In Ref.~\cite{HO17}, the results for the Stein exponent are sharpened up to the third-order term. Moreover, the Hoeffding exponents for the two-way LOCC and SEP POVMs coincide. There is a difference in the Hoeffding exponent between the one-way LOCC and two-way LOCC POVMs unless the Schmidt coefficients are uniform.

Below we highlight the differences between our work from Refs.~\cite{OH14, OH15, HO17}. 
First, since the pair of  states for testing in \cite{OH14, OH15, HO17} are not orthogonal, the faithfulness of the error exponents therein holds trivially.
Namely, the discrimination error between an entangled bipartite state and the completely mixed state can not decay super-exponentially.
On the other hand, testing orthogonal states in our work becomes highly nontrivial because examples show that more than one copy is required for perfect discrimination. In contrast, two copies suffice for discrimination without mistake (see e.g.~\cite{YDY12, YDY14}).
Second, the proofs in Refs.~\cite{OH14, OH15, HO17} heavily rely on the equivalence of testing a bipartite pure state against the white noise state to a global hypothesis-testing problem between a \emph{classical distribution} against the uniform one. Note that the white noise state plays a crucial role here. It is unclear whether such an equivalence still holds generally.
As for the first part of this work, we primarily rely on the twirling operation of the hypotheses.
Third, Refs.~\cite{OH14, OH15, HO17} concerned hypothesis testing using SEP POVMs and reduced it to a global hypothesis testing with a composite null hypothesis. 
This work, however, does not directly analyze the SEP POVMs but the PPT POVMs. We remark that a converse bound for testing using SEP POVMs does not necessarily yield a converse bound for testing using PPT POVMs since $\text{SEP}\subset \text{PPT}$. Moreover, Section~\ref{sec:separation} establishes an infinite separation between $\text{SEP}$ and $\text{PPT}$, which seems novel to our knowledge.
Lastly, the second-order expansion of the optimal type-II error with type-I no larger than a constant studied in \cite{HO17} generally resembles the classical case \cite{TH13, Li14}. For  testing the maximally entangled state against the white noise state, \cite[Theorem 4]{HO17} showed that, for `X' being LOCC or SEP, $\beta_{\varepsilon}^{\mathrm{X}}(\Phi_d^{\otimes n}, (\mathds{1}_d/d)^{\otimes n} ) = (1-\varepsilon) / d$, which is similar to ours given in \eqref{eq:trade-off0} probably because of symmetry of the maximally entangled state.
Although we are not concerned with testing against a white noise state in this paper, our analysis of the PPT-distinguishability (Proposition~\ref{prop:Exp_LB} in Section~\ref{sec:converse_sym}) shows that the above result holds for PPT POVMs as well, which strengthens  the result \cite[Theorem 4]{HO17} from SEP to PPT POVMs.

\section{Notation and Definitions} \label{sec:notation}
Let $\mathbb{C}^d$ be a $d$-dimensional complex Hilbert space. A quantum state (i.e.,~density operator) 
on $\mathbb{C}^d$ is a positive semi-definite operator with unit trace. 
The operator `$\dagger$' denotes complex conjugate and transpose.
The trace operation is 
denoted by $\Tr[\,\cdot\,]$.
The symbol $\mathds{1}_{d}$ stands for the identity operator on $\mathbb{C}^{d}$, and $\mathbb{O}_{d}$ 
denotes the zero operation on $\mathbb{C}^{d}$. If there is no ambiguity, we will drop the subscript for simplicity. 
For a bounded operator $M$, we denote by $|M|:= \sqrt{M^\dagger M}$ its absolute value.
We use $\|M\|_\infty$ for the usual operator norm of $M$ (i.e.~the largest of the eigenvalue of $|M|$), and $\|M \|_1 := \Tr[|M|]$ for the trace norm (i.e.~the sum of the singular values). 

A positive-operator valued measure (POVM) is a set of positive 
semi-definite operators whose sum equals identity.
For any density operator $\rho$ on $\mathbb{C}^d$, we use $\rho^\perp$ to denote the \emph{orthogonal complement} of $\rho$, which is defined as the density  
operator on $\mathbb{C}^d$ that is maximally mixed on the orthogonal complement of the support of $\rho$. 
For example, if $\rho$ is pure (i.e.~a rank-one projection), then $\rho^\perp = (\mathds{1}_d - \rho)/(d-1)$.
The maximally entangled state on $\mathbb{C}^d \otimes \mathbb{C}^d$ is denoted by 
$\Phi_d := \frac1d \sum_{i,j=0}^{d-1} |ii\rangle\langle jj| $, and $\Phi_d^\perp:= (\mathds{1} - \Phi_d)/(d^2-1)$ is its 
orthogonal complement.
We use $\tau_d := \mathds{1}_d / d$ to denote the completely mixed state on $\mathbb{C}^d$.
The operations $\otimes $ and $\oplus$ represent the tensor product and direct sum, respectively.
We use $\mathbb{N}$ to denote natural numbers.
For an operator on a bipartite Hilbert space $\mathcal{H}_A\otimes \mathcal{H}_B$, 
we use $\Gamma$ to denote its partial transpose with respect to $\mathcal{H}_B$, i.e.,~for 
some orthonormal basis $\{|i\rangle_A\otimes |j\rangle_B\}_{(i,j)}$, we define the partial transpose as
\begin{align}
	\left( |i\rangle\langle k |_A \otimes |j\rangle\langle \ell |_B \right)^\Gamma := |i\rangle\langle k |_A\otimes |\ell\rangle   \langle j |_B.
\end{align}
We remark that the partial transpose $\Gamma$ is an isometry with respect to the Hilbert--Schmidt inner product, i.e., $\Tr[A B^\Gamma] = \Tr[A^\Gamma B]$ for Hermitian matrices $A$ and $B$.
We will use this fact throughout the paper.

Given a multipartite system $\mathcal{H}_1\otimes \cdots \otimes  \mathcal{H}_m$, a (one-way) 
LOCC POVM \cite{CLM+14} is a decision rule based on all the measurement outcomes performed 
locally on each subsystem. The SEP measurements are defined as
\begin{align}
	\textrm{SEP} :=
	\left\{
	\left( E_k^{(1)} \otimes \cdots \otimes E_k^{(m)} \right)_k:
	E_k^{(j)} \geq 0, \;
	\sum_k E_k^{(1)} \otimes \cdots \otimes E_k^{(m)} = \mathds{1}
	\right\},
\end{align}
and the PPT POVMs are defined as 
\begin{align}
	\textrm{PPT} :=
	\left\{
	\left( E_k \right)_k \; \text{POVM}: \forall 1\leq j\leq m, \forall k, 
	\left(  E_k \right)^\Gamma  \geq 0,
	\text{ where $\Gamma$ is partial transpose on $\mathcal{H}_j$.}
	\right\}.
\end{align}
We remark that for the LOCC measurements, classical post-processing, i.e.,~a decision rule, at the end is allowed.

Consider a binary hypothesis testing problem as follows:
\begin{align}
	\begin{cases}
		\mathsf{H}_0: \rho_0^{\otimes n} \\
		\mathsf{H}_1: \rho_1^{\otimes n}
	\end{cases}, \forall n\in\mathbb{N}.
\end{align}
Given a test $T$ where $(T,\mathds{1}-T)$ forms a two-outcome POVM, we let 
$\alpha_n(T):= \Tr\left[\rho^{\otimes n}(\mathds{1} - T) \right]$ and 
$\beta_n(T) := \Tr\left[ \sigma^{\otimes n} T \right]$.
In the symmetric case with prior $0<p<1$, we define the optimal error probability 
of the binary hypothesis testing using the class X of POVMs as:
\begin{align}
P_\text{e}^\text{X}(\rho_0^{\otimes n}, \rho_1^{\otimes n}; p ) := \inf_{  T, (\mathds{1}-T) \in \text{X}  } \left\{  p \alpha (T_n) + (1-p) \beta (T_n) \right\}. \label{eq:optimal_err_X}
\end{align}
The associated Chernoff exponent is expressed as
\begin{align} \label{eq:Chernoff_X}
	\mathrm{Chernoff}^\text{X}(\rho_0, \rho_1) &:= \lim_{n\to \infty} - \frac1n \log P_\text{e}^\text{X} (\rho_0^{\otimes n}, \rho_1^{\otimes n}; p ).
\end{align}
We note that the definition given in Eq.~\eqref{eq:optimal_err_X} can be naturally extended to multiple hypothesis testing with priors.

Given density matrices $\rho_0$ and $\rho_1$, we use the following definitions:
\begin{align}
	\mathcal{R}^\text{X}(\rho_0, \rho_1)
	 &
	= \left\{(\alpha,\beta) : \exists \, T, \mathds{1}-T \in \text{X}: \alpha = \Tr\left[ (\mathds{1} - T) \rho_0 \right], \; \beta = \Tr\left[ T \rho_1 \right]
	\right\},
\\
	\alpha_\beta^\text{X}(\rho_0, \rho_1)  &:= \inf \left\{ \alpha: (\alpha,\beta) \in \mathcal{R}^\text{X}(\rho_0, \rho_1) \right\}, \label{eq:trade-off_alpha}\\
	\beta_\alpha^\text{X}(\rho_0, \rho_1)  &:= \inf \left\{ \beta: (\alpha,\beta) \in \mathcal{R}^\text{X}(\rho_0, \rho_1) \right\}. \label{eq:trade-off}
\end{align}
Here, $\mathcal{R}^\text{X}(\rho_0, \rho_1)$ is sometimes called the \emph{hypothesis-testing region}, and the trade-off between type-I and type-II errors, $\beta_\alpha^\text{X}(\rho_0, \rho_1)$, is termed as the \emph{Neyman--Pearson function} or the \emph{trade-off function} (see ~\cite[Section 3.2]{Leh86}, \cite{Pol13}, \cite[Definition 2.1]{DRS+19}).

In the asymmetric case, we define the following exponent functions:
\begin{align}
	\mathrm{Stein}^\text{X} (\rho_0, \rho_1, \eps) &:= \lim_{n\to \infty} \sup_{ T, (\mathds{1}-T) \in \text{X}    }  \left\{ -\frac1n \log \beta_n(T) : \alpha_n (T) \leq \eps  \right\}, \quad \forall \eps \in (0,1); \label{eq:Stein}\\
	\mathrm{Hoeffding}^\text{X} (\rho_0, \rho_1, r) &:= \lim_{n\to \infty} \sup_{ T, (\mathds{1}-T) \in \text{X}   }  \left\{ -\frac1n \log \alpha_n(T) : -\frac1n \log \beta_n (T) \geq r  \right\}, \quad \forall r > 0; \label{eq:Hoeffding} \\
	\mathrm{SC}^\text{X} (\rho_0, \rho_1, r) &:= \lim_{n\to \infty} \sup_{ T, (\mathds{1}-T) \in \text{X}    }  \left\{ -\frac1n \log ( 1-\alpha_n(T) ) : -\frac1n \log \beta_n (T) \geq r  \right\}, \quad \forall r > 0. \label{eq:sc}
\end{align}

\section{Testing Arbitrary Entangled Pure States} \label{sec:general}

This section is devoted to proving that the Chernoff bound of testing an arbitrary 
multipartite entangled pure state against its orthogonal complement using LOCC POVMs 
is faithful. Namely, the associated optimal error probability decays exponentially in 
the number of copies (see Theorem~\ref{theo:faithfulness_Chernoff}).

The main technique in establishing this result is to show an exponential lower bound 
of the optimal error probability using PPT POVMs, as stated in
Proposition~\ref{prop:Exp_LB} of Section~\ref{sec:converse_sym}.
Moreover, our technique also extends to provide a strong converse bound to the 
Stein exponent and a lower bound to the strong converse exponent in the asymmetric 
setting (Section~\ref{sec:converse_sym}).

\begin{theo}[Faithfulness of the Chernoff bound] \label{theo:faithfulness_Chernoff}
	Let $\psi $ be an arbitrary multipartite entangled pure state  on $\mathbb{C}^{d_1}\otimes \mathbb{C}^{d_2}\otimes \cdots\otimes \mathbb{C}^{d_m}$
	and $\psi^\perp $ be its orthogonal complement state.
	The optimal error probability of discriminating them using LOCC POVMs decays exponentially in the number of copies, i.e.,~for any $0<p<1$ there exists $0< b<\infty$ such that
	\begin{align}
		\e^{-nb}	\leq 	P_\textnormal{e}^\textnormal{LOCC}\left( \psi^{\otimes n} , (\psi^\perp)^{\otimes n}; p \right) \leq 
			(1-p)\left( 1-\frac{d_{m-1}d_m-\min\{d_{m-1},d_m\}}{d_1d_2\cdots d_m-1} \right)^n, \quad \forall n\in\mathbb{N}.
	\end{align}
	In other words,
	\begin{align}
		0< 1-\frac{d_{m-1}d_m-\min\{d_{m-1},d_m\}}{d_1d_2\cdots d_m-1}  \leq \mathrm{Chernoff}^\textnormal{LOCC} \left( \psi , \psi^\perp\right) \leq b < \infty.
	\end{align}
\end{theo}

\begin{remark}
	The lower bound to $P_\textnormal{e}^\textnormal{LOCC}\left( \psi^{\otimes n} , (\psi^\perp)^{\otimes n}; p \right)$ relies on the converse bound for the bipartite scenario that will be shortly introduced in Proposition~\ref{prop:Exp_LB} of Section~\ref{sec:converse_sym}.
\end{remark}

\begin{remark}
		As we will prove later in Theorem~\ref{theo:MES} of Section~\ref{sec:MES}, the upper bound 
		$P_\textnormal{e}^\textnormal{LOCC}\left( \psi^{\otimes n} , (\psi^\perp)^{\otimes n}; p \right) \leq \frac{1-p}{(d+1)^n}$
		given in Theorem~\ref{theo:faithfulness_Chernoff} is tight when $\psi$ is a bipartite maximally entangled on $\mathds{C}^d\otimes \mathds{C}^d$.
\end{remark}

\begin{proof}[Proof]

	We first prove the upper bound.
	Observe that the pure state $\ket{\psi}$ can be expanded in the computational basis of the first $m-2$ system
	\begin{align*}
		\ket{\psi}=\ket{0\cdots 0}\ket{\phi}_{m-1,m}+\sum_{j\neq 0\cdots 0} \ket{j}\ket{ \phi_j}_{m-1,m}.
	\end{align*}
	Let $\sum_{i=0}^{d-1} \sqrt{\lambda_i} |ii\rangle$ be the Schmidt decomposition 
	of $ |\phi\rangle_{m-1,m}$ for some Schmidt basis $\{|i\rangle\}_{i=0}^{d-1}$, where $d\leq d_{m-1},d_m$.
	For each copy, we measure the first $m-2$ systems in the computational basis and the last two systems in 
	the Schmidt basis corresponding to $ |\phi\rangle_{m-1,m}$ via an LOCC protocol. The resulting probability distributions of measuring $\psi$ is $P$ and the probability distribution of measuring $\psi^\perp$ is
	\begin{align*}
		Q=\frac{d_1d_2\cdots d_m}{d_1d_2\cdots d_m-1}U- \frac{1}{d_1d_2\cdots d_m-1}P
	\end{align*}
	where $U$ is the uniform distribution over $\{0,1,\ldots, d_1d_2\cdots d_m-1 \}$, and the support of $P$ is at most $d_1d_2\cdots d_{m}+d-d_{m-1}d_m$.
	Using the classical Chernoff bound we obtain, for all $\alpha\in[0,1]$,
	\begin{align}
		P_\textnormal{e}^\textnormal{LOCC}\left( \psi^{\otimes n} , (\psi^\perp)^{\otimes n}; p \right) 
		&\leq 
		p^\alpha (1-p)^{1-\alpha}  \left(\sum_{x} \left(P(x)\right)^\alpha \left(Q (x)\right)^{1-\alpha} \right)^n\\
		&\leq
		(1-p)\left( 1-\frac{d_{m-1}d_m-\min\{d_{m-1},d_m\}}{d_1d_2\cdots d_m-1} \right)^n, \quad n\in\mathbb{N},
	\end{align}
	where we simply chose $\alpha = 0$ in the last line to conclude the upper bound.
	
	For the lower bound, note that any multipartite system can be viewed as a bipartite system by grouping the $m$ parties into two nonempty sets. Since we assume $\psi$ to be multi-party entangled, there exists a bipartition with respect to which $\psi$ is entangled (between the two groups). 
	Hence, LOCC POVMs on the multipartite system are contained in the LOCC POVMs on the associated bipartite system. The minimum error probability of the latter is thus a lower bound on the error probability of the former. Then, it suffices to apply the exponential lower bound to the error probability using LOCC POVMs on a bipartite system to complete the proof, as shown in Proposition~\ref{prop:Exp_LB} later in Section~\ref{sec:converse_sym}.
\end{proof}

\begin{remark}
Let us emphasize that if $\psi$ is not entangled, but rather a product state, then 
the distribution $P$ is a singleton, and hence the Chernoff exponent becomes $\infty$ by choosing $\alpha=1/2$.
\end{remark}

\subsection{Exponential Strong Converse Bound for Bipartite Pure States} \label{sec:converse_sym}

The faithfulness of testing arbitrary multipartite entangled state in Theorem~\ref{theo:faithfulness_Chernoff} relies on an exponential lower bound to $P_\text{e}^\text{LOCC}$.
In this section, we present our main proof technique---an upper bound on the PPT-distinguishability norm---in Lemma~\ref{lemm:PPT} below.
Then we will show how it gives exponential converse bounds to the error probability using PPT POVMs (Proposition~\ref{prop:Exp_LB}), which implies an lower bound for LOCC POVMs.
	
We define the PPT-distinguishability norm  \cite{MWW09} for any Hermitian matrix $H$ as
\begin{align}
	\left\|H\right\|_\text{PPT} &:= \sup_{ -\mathbb{1} \leq M, M^\Gamma \leq \mathbb{1} } \Tr\left[ H M \right]. \label{eq:primal_PPT}
\end{align}	
	
\begin{lemm}[Upper bound on PPT-distinguishability norm] \label{lemm:PPT}
	Let $\rho_0$ and $\rho_1$ be arbitrary bipartite density matrices.
	Provided that
	\begin{align} \label{eq:PPT_constraint0}
	-\rho_1^\Gamma \leq t \cdot \rho_0^\Gamma	\leq \rho_1^\Gamma
	\end{align}
	 for some $t \in [0,1]$, then
	\begin{align}
		\|\rho_0^{\otimes n} - \lambda \rho_1^{\otimes n} \|_{\textnormal{PPT}} \leq |1-\lambda t^n | + \lambda(1-t^n), \quad \forall\, \lambda\geq 0.
	\end{align}
\end{lemm}

\begin{remark}
	The condition \eqref{eq:PPT_constraint0} obviously holds for $t=0$ with any PPT state $\rho_1$, but it would only yield a trivial bound:
	$\|\rho_0^{\otimes n} - \lambda \rho_1^{\otimes n} \|_{\textnormal{PPT}} \leq \|\rho_0^{\otimes n} - \lambda \rho_1^{\otimes n} \|_{1} \leq 1+\lambda$.
\end{remark}

\begin{proof}
Our key technique is to succinctly formulate the PPT-distinguishability norm \cite{MWW09} into its dual representation in terms of the Schatten $1$-norm. Then, we find a feasible solution to the dual problem to prove our claim.	
We first derive a dual program to the PPT-distinguishability norm given in \eqref{eq:primal_PPT}.

Since there are four linear inequality constraints in the above primal optimization, we write the Lagrangian $L$ by introducing the corresponding dual variables $A,B,C,D \geq 0$ as
	\begin{align}
		L &= \Tr\left[H M\right] 
		+ \Tr\left[ A(\mathbb{1}-M) \right] + \Tr\left[ B(\mathbb{1}+M) \right] + \Tr\left[ C(\mathbb{1}-M^\Gamma) \right] + \Tr\left[ D(\mathbb{1}+M^\Gamma) \right] \\
		&= \Tr\left[ A+B+C+D \right] + \Tr\left[ M \left( H-A+B-C^\Gamma+D^\Gamma \right) \right].
	\end{align}
	To derive the dual program of Eq.~\eqref{eq:primal_PPT}, we maximize $L$ over all Hermitian matrices $M$ without constraints. Clearly, this is infinite unless $H = A-B+(C-D)^\Gamma$. For such a case, the maximum is simply $\Tr\left[ A+B+C+D \right]$.
	The dual formulation for the PPT-distinguishability norm is
	\begin{align}
		\inf_{A,B,C,D\geq 0} \left\{ \Tr\left[ A+B+C+D \right] : H = A-B + (C-D)^\Gamma
		\right\} .
	\end{align}
	It is easy to see that a strictly feasible solution exists for the dual program; the strong duality holds for the semidefinite program considered here.
	
	Moreover, we can rewrite the above in terms of $1$-norm by denoting Hermitian matrices $X = A-B$ and $Y = (C-D)^\Gamma$ as follows:
	\begin{align} \label{eq:dual_PPT_norm}
		\left\| H \right\|_\text{PPT} &= \inf_{X = X^\dagger,\, Y = Y^\dagger} \left\{ \left\| X \right\|_1 + \left\| Y^\Gamma \right\|_1 : H = X + Y
		\right\}.
	\end{align}
	This is because the optimal decomposition of the absolute value of a matrix is into its positive and negative parts. That is, 
	\begin{align}
		\left\| X \right\|_1 &= \sup_{-\mathbb{1}\leq M \leq \mathbb{1}} \Tr\left[ M X \right] \\
		&= \inf_{A,B\geq 0} \sup_{ M } \left\{ \Tr[MX] + \Tr\left[A\left( \mathbb{1}-M\right)\right] + \Tr\left[B\left( \mathbb{1}+M \right)\right] \right\} \\
		&= \inf_{A,B\geq 0} \left\{ \Tr\left[A+B\right] : X = A-B\right\}.
	\end{align}

	Next, 
	we find a feasible solution to the dual program, \eqref{eq:dual_PPT_norm}, i.e.~ 
	\begin{align}
		H &= \rho_0^{\otimes n} - \lambda \rho_1^{\otimes n},\\
		X &= (1- \lambda t^n) \rho_0^{\otimes n}, \\
		Y &= \lambda \cdot \left( t^n \rho_0^{\otimes n} - \rho_1^{\otimes n} \right).
	\end{align}
	To obtain an upper bound to the dual program, we calculate:
	\begin{align}
		Y^\Gamma &= \lambda \cdot \left( \left( t \rho_0^\Gamma \right)^{\otimes n} - (\rho_1^\Gamma)^{\otimes n} \right) 
		\leq 0,
	\end{align}
	where the operator inequality follows from the hypothesis \eqref{eq:PPT_constraint0} and Lemma~\ref{lemm:tensor} given below.

	Hence, we obtain an upper bound to the dual program:
	\begin{align}
			\| \rho_0^{\otimes n} - \lambda \rho_1^{\otimes n} \|_{\textnormal{PPT}} 
			\leq \|X\|_1 + \left\|Y^\Gamma \right\|_1
			= \Tr[|X|] - \Tr\left[Y^\Gamma\right] = \Tr[|X|] - \Tr[Y]
			= |1-\lambda t^n | + \lambda(1-t^n),
	\end{align}
	concluding the proof.
\end{proof}

\begin{lemm}\label{lemm:tensor}
	For self-adjoint operators $A$ and $Z$ satisfying
	$-A\leq Z \leq A$, then
	\begin{align} \label{eq:tensor}
		-A^{\otimes n} \leq Z^{\otimes n} \leq A^{\otimes n}, \quad \forall\, n\in \mathds{N}.
	\end{align}
\end{lemm}

\begin{proof}
	First, one can assert that $A$ is positive semi-definite because $A\geq -A$ is equivalent to $2 A\geq 0$.
	Furthermore, both the supports of the positive part of $Z$ and the negative parts of $Z$ must be contained in the support of $A$. Otherwise, the hypothesis $-A\leq Z \leq A$ cannot be true.
	Hence, we have
	\begin{align}
		-\mathds{1} \leq A^{-\frac12} Z A^{-\frac12} \leq \mathds{1},
	\end{align}
	where the inverse of $A$ is taken on the support of $A$.
	Since the eigenvalues of $A^{-\frac12} Z A^{-\frac12}$ is bounded in $[-1,1]$, the eigenvalues of the tensor power are also bounded in $[-1,1]$ via multiplication.
	Namely, we obtain
	\begin{align}
		-\mathds{1}^{\otimes n} &\leq \left(A^{-\frac12} Z A^{-\frac12} \right)^{\otimes n}
		\\
		&= \left( A^{\otimes n} \right)^{-\frac12} Z^{\otimes n} \left( A^{\otimes n} \right)^{-\frac12} 
		\\
		&\leq \mathds{1}^{\otimes n},
	\end{align}
	which, in turn, implies the desired inequality in \eqref{eq:tensor}.
\end{proof}

If \eqref{eq:PPT_constraint0} in Lemma~\ref{lemm:PPT} holds for some positive factor $t>0$, then such a factor will play a crucial role in providing a fundamental limit, for which the error probability under the PPT measurements cannot decrease too fast (i.e., the converse bound).

\begin{prop}[Exponential lower bounds for PPT POVMs] \label{prop:Exp_LB}
Let $\rho_0$ and $\rho_1$ be arbitrary bipartite density matrices.
Provided that \eqref{eq:PPT_constraint0} holds for some $t\in(0,1]$,  the error probability of testing multi-copies of $\rho_0$ and $\rho_1$
using LOCC POVMs is lower bounded by
\begin{align} \label{eq:PPT_lower_bound}
	P_\textnormal{e}^\textnormal{LOCC}\left( \rho_0^{\otimes n} , \rho_1^{\otimes n}; p \right) \geq P_\textnormal{e}^\textnormal{PPT}\left( \rho_0^{\otimes n} , \rho_1^{\otimes n}; p \right)  
	\geq \min\{(1-p) \cdot t^n, p\} 
	\geq \min\{1-p, p \} \cdot t^n, \quad \forall n\, \in\mathbb{N}
\end{align}
Additionally, for any $r> \log \tfrac{1}{t} $, we have
\begin{align}
	\alpha_{\e^{-nr}}^\textnormal{LOCC}\left( \rho_0^{\otimes n} , \rho_1^{\otimes n} \right) \geq
	\alpha_{\e^{-nr}}^\textnormal{PPT}\left( \rho_0^{\otimes n} , \rho_1^{\otimes n} \right) \geq 1 - \e^{-n\left[ r - \log \frac{1}{t} \right]}, \quad \forall n\in\mathbb{N},
\end{align}	
where the optimal type-I error $\alpha_{\mu}^\textnormal{X}$ is defined in Eq.~\eqref{eq:trade-off_alpha} of Section~\ref{sec:notation}.

In particular for the case of $ \rho_0 = \psi $ being a bipartite entangled pure state on $\mathbb{C}^d\otimes \mathbb{C}^d$ and $ \rho_1 = \psi^\perp := \tfrac{\mathbb{1}-\psi}{d^2-1}$ being its orthogonal complement, then
$t = \tfrac{1-\eta}{(d^2-1) \eta } \in (0,\tfrac{1}{d+1}]$ holds in \eqref{eq:PPT_constraint0}, where $\eta$ denotes the largest squared Schmidt coefficient of $\psi$, i.e.,~$\eta := \|\Tr_A[\psi]\|_\infty \in [\tfrac{1}{d}, 1)$.
\end{prop}

\begin{remark}
	Proposition~\ref{prop:Exp_LB} applies to testing general bipartite quantum states, albeit we only focus on testing orthogonal states in this paper.
	Furthermore, as mentioned in Theorem~\ref{theo:faithfulness_Chernoff}, the converse bound for testing bipartite states considered in Proposition~\ref{prop:Exp_LB} implies a converse bound for testing arbitrary multipartite states.
	Hence, Proposition~\ref{prop:Exp_LB} provides a general method to witness the faithfulness of a Chernoff bound. 
	Later in Section~\ref{sec:optimal_LOCC}, we show that Proposition~\ref{prop:Exp_LB} is actually tight for certain scenarios.
\end{remark}

\begin{proof}[Proof]

We first prove the claim for symmetric hypothesis testing.
We apply Lemma~\ref{lemm:PPT} with $\lambda = \frac{1-p}{p} \geq 0$ to lower bound the error probability under PPT POVM:
	\begin{align}
		P_\text{e}^\text{PPT} \left( \rho_0^{\otimes n} , \rho_1^{\otimes n};  p \right) &= 
		\inf_{T, (\mathbb{1}-T) \in \text{PPT} } \;  p \Tr\left[ \psi^{\otimes n} (\mathbb{1}-T)\right] + (1-p) \Tr\left[ (\psi^\perp)^{\otimes n}  T \right] \\
		&= 
		\inf_{T, (\mathbb{1}-T) \in \text{PPT} } \; \frac12 + \Tr\left[\left( p \rho_0^{\otimes n} - (1-p) \rho_1^{\otimes n} \right)\left(2T-\mathbb{1}\right)\right]\\
		&=
		\frac12 \left( 1 - \left\| p \rho_0^{\otimes n} - (1-p) \rho_1^{\otimes n} \right\|_\text{PPT}\right) \\ \label{eq:PPT_error}
		&\geq \frac12 \left( 1 - \left| p - (1-p)\cdot t^n \right| - (1-p) \cdot \left(1 - t^n \right) \right) \\
		&= \frac12 \left( p + (1-p) t^n - \left| p - (1-p) \cdot t^n \right| \right) \\
		&=  \min\left\{ (1-p) \cdot t^n , p 
		\right\},
	\end{align}
yielding the first claim.

For asymmetric hypothesis testing, we assert that, for all $\mu\geq 0$,
\begin{align} \label{eq:alpha_variational0}
	\alpha_{\mu}^\textnormal{{PPT}} \left(\rho, \sigma \right) = \max_{\lambda\geq 0} \left\{ \frac12\left( 1+\lambda - \left\| \rho - \lambda \sigma \right\|_\textnormal{PPT} \right) 
	- \lambda \mu \right\}.
\end{align}
This is proved by the following:
\begin{align}
	\alpha_{\mu}^\textnormal{{PPT}} \left(\rho, \sigma \right) &:= \inf_{ T\in \textnormal{{PPT}} } \left\{ \Tr\left[\rho(\mathds{1}- T)  \right] : \Tr\left[\sigma T \right] \leq \mu \right\} \\
	&= \inf_{ T\in \textnormal{{PPT}} } \max_{\lambda\geq 0 }  \left\{ \Tr\left[\rho(\mathds{1}- T)  \right]  + \lambda \Tr\left[\sigma T \right] - \lambda \mu	\right\} \\
	&\overset{ \textnormal{(a)}}{=} \max_{\lambda\geq 0 } \inf_{ T\in \textnormal{{PPT}} } \left\{ \Tr\left[\rho(\mathds{1}- T)  \right]  + \lambda \Tr\left[\sigma T \right] - \lambda \mu	\right\} \\
	&\overset{}{=} \max_{\lambda\geq 0 } \inf_{ T\in \textnormal{{PPT}} } \left\{ \frac12\left( 1+\lambda - \Tr \left[ (\rho - \lambda \sigma) ( 2T - \mathds{1}  ) \right] \right) - \lambda \mu \right\} \label{eq:alpha_variational1} \\			
	&\overset{\textnormal{(b)}}{=} \max_{\lambda\geq 0} \left\{ \frac12\left( 1+\lambda - \left\| \rho - \lambda \sigma \right\|_\textnormal{{PPT}} \right) 
	- \lambda \mu \right\}, \label{eq:alpha_variational2} 
\end{align}
where in (a) we 
change the order of inf and max by Sion's minimax theorem and noting that the objective function is linear in $\lambda$ and $T$, respectively;
and (b) follows from the PPT-distinguishability norm given in Eq.~\eqref{eq:primal_PPT}.

Application of Lemma~\ref{lemm:PPT} with $\lambda = t^{-n}$ to \eqref{eq:alpha_variational2} leads us to
\begin{align}
	\alpha_{\exp\{-nr\}}^\textnormal{PPT}\left( \psi^{\otimes n} , (\psi^\perp)^{\otimes n} \right) &\geq 
	\frac12\left[ 1+\lambda  - |1 - \lambda t^n |  -  \lambda(1-t^n)  \right] - \lambda \e^{-nr}  \\
	&= 1 - t^{-n} \e^{-nr},
\end{align}
which shows the second claim.

Lastly, to see  $t =  \frac{1-\eta }{(d^2-1) \eta  } \in ( 0, \frac{1}{d+1}]$ holds for \eqref{eq:PPT_constraint0} with $\eta = \|\Tr_A[\psi]\|_\infty \in [\tfrac{1}{d}, 1)$, we verify the following stronger condition:
\begin{align}
	(\psi^\perp)^\Gamma - t \cdot  \left|\psi^\Gamma \right|  \geq 0,
\end{align}
which implies that \eqref{eq:PPT_constraint0} must be fulfilled.
Namely,
\begin{align}
	(\psi^\perp)^\Gamma - t \cdot  \left|\psi^\Gamma \right|  
	&= \frac{1}{d^2-1} \left( \mathds{1} - \psi^\Gamma - \frac{1-\eta}{\eta} \left|\psi^\Gamma \right|  \right) \\
	&\geq \frac{1}{d^2-1} \left( \mathds{1} - \left|\psi^\Gamma\right| - \frac{1-\eta}{\eta} \left|\psi^\Gamma \right|  \right) \\
	&= \frac{1}{d^2-1} \left( \mathds{1} - \frac{1}{\eta} \left|\psi^\Gamma \right|  \right) \\
	&\geq 0.
\end{align}
Here, the last inequality holds because  $\psi^\Gamma$ has eigenvalues $\pm\sqrt{\lambda_i \lambda_j}$ for $i\neq j$ and $\lambda_i$, where $\{\lambda_i\}_{1\leq i\leq d}$ denote the squared Schmidt coefficients of $\psi$ \cite[Lemma 14]{YDY14}; 
then the eigenvalues of $\frac{1}{\eta} \left|\psi^\Gamma \right| $ are $\frac{\sqrt{\lambda_i \lambda_j}}{ \max_{1\leq k\leq d} \lambda_k  } \leq 1$.
This concludes the proof.

%
\end{proof}

Proposition~\ref{prop:Exp_LB} yields the faithfulness of the Chernoff exponent for distinguishing arbitrary pure bipartite entangled state against its orthogonal complement.
Moreover, the type-I error converges to $1$ exponentially fast as long 
as the exponential decay rate of the type-II error exceeds $\log \tfrac{1}{t}$. 
This result gives a strong converse bound for the Stein exponent of distinguishing 
$ \psi^{\otimes n}$ against $(\psi^\perp)^{\otimes n}$.

\begin{coro}[Strong converse bound for the Stein exponent] \label{coro:general_sc}
	Consider any pure bipartite entangled state $\psi$ on $\mathbb{C}^d\otimes \mathbb{C}^d$ and its orthogonal complement $\psi^\perp := \tfrac{\mathbb{1}-\psi}{d^2-1}$.	
	Let $t:= \frac{1-\eta }{(d^2-1) \eta  } \in ( 0, \frac{1}{d+1}]$ and $\eta := \|\Tr_A[\psi]\|_\infty \in [\tfrac{1}{d}, 1)$.	
	Then,
		\begin{align}
		&\mathrm{Chernoff}^\textnormal{LOCC} \left( \psi , \psi^\perp\right) 
		\leq\mathrm{Chernoff}^\textnormal{PPT} \left( \psi , \psi^\perp\right)
		\leq \log \tfrac{1}{t}	< \infty;
		\\			
		&\mathrm{SC}^\textnormal{PPT}\left( \psi , \psi^\perp, r \right) \geq r - \log \tfrac{1}{t}, \quad \forall r> \log \tfrac{1}{t};
		\\
		&\mathrm{Stein}^\textnormal{PPT}\left( \psi^{\otimes n} , (\psi^\perp)^{\otimes n}, \varepsilon \right) \leq \log \tfrac{1}{t} \in \left[ \log (d+1), \infty \right), \quad \forall \varepsilon\in(0,1).
	\end{align}
\end{coro}

\section{Optimal LOCC Protocols for Binary Hypothesis Testing} \label{sec:optimal_LOCC}


In Section~\ref{sec:general}, we established exponential lower bounds to the optimal errors using PPT POVMs. This section will show a sufficient condition under which such lower bounds are achievable by (one-way) LOCC POVMs (Theorem~\ref{theo:sufficient_condition}). Namely, it characterizes a sufficient condition when the three classes, LOCC, SEP, and PPT, collapse for testing orthogonal bipartite states. In Section~\ref{sec:MES} and Section~\ref{sec:Werner}, we will show that the condition holds for testing a maximally entangled state against its orthogonal complement and testing extremal Werner states. Hence, the corresponding optimal LOCC protocols in symmetric and asymmetric hypothesis testing are obtained.

We start with an achievability result.
\begin{prop}[Achievability] \label{prop:achievability}
	Let $\rho_0$ and $\rho_1$ be arbitrary orthogonal states,
	and let
		$\Pi_{\rho_i}$ denotes the projection onto the support of $\rho_i$.
	
	If there exists a $t\in(0,1]$ satisfying
	\begin{align} \label{eq:Achivability_condition}
		\left\{ M = t \cdot \Pi_{\rho_1} + \Pi_{\rho_0}, \mathds{1} - M  = (1-t) \cdot \Pi_{\rho_1}   \right\} \in \textnormal{X},
	\end{align}
	then, for any $p\in(0,1)$ and $n\in\mathds{N}$,
	\begin{align}
		P_\textnormal{e}^\textnormal{X} \left( \rho_0^{\otimes n}, \rho_1^{\otimes n}; p \right) \leq \min\left\{ (1-p)\cdot t^n, p \right\},
	\end{align}
	and, for any $\alpha \in [0,1]$,
	\begin{align}
		\beta_\alpha^\textnormal{X} \left( \rho_0^{\otimes n}, \rho_1^{\otimes n} \right) \leq (1-\alpha) \cdot t^n.
	\end{align}
	Here, 
	`\textnormal{X}' means a restricted class of measurements.
\end{prop}
\begin{proof}
	For symmetric hypothesis testing, we choose $\{M^{\otimes n}, \mathds{1} - {M}^{\otimes n}\}$ for testing by the assumption in \eqref{eq:Achivability_condition}.
	Note that from \eqref{eq:Achivability_condition},
	\begin{align}
		M^{\otimes n} = \sum_{k=0}^n t^{n-k} \Lambda_k,
	\end{align}
	where $\Lambda_k$ denotes the sum of all elements of $\{ \Pi_{\rho_0}, \Pi_{\rho_1}  \}^{\otimes n}$ that have $k$ copies of $\Pi_{\rho_0}$.
	Then,
	\begin{align}
		P_\textnormal{e}^\textnormal{X} \left( \rho_0^{\otimes n}, \rho_1^{\otimes n}; p \right) 
		&\leq p \Tr\left[ (\mathds{1}- M^{\otimes n}) \rho_0^{\otimes n} \right] + (1-p) \Tr\left[ M^{\otimes n} \rho_1^{\otimes n} \right] \\
		&= p \Tr\left[ (\mathds{1}- \Lambda_n) \rho_0^{\otimes n} \right] + (1-p) \Tr\left[ \Lambda_0 \rho_1^{\otimes n} \right] \cdot t^n \\
		&= (1-p) \cdot t^n.
	\end{align}
	
	On the other hand, if the above upper bound is strictly greater than $p$, we simply choose measurement $\{ \mathds{O}, \mathds{1}\}$, i.e., the test always opts for $\mathsf{H}_1$.
	This shows the first claim.
	
	Next, we move on to asymmetric hypothesis testing.
	For every $\alpha \in [0,1]$, we choose the test
	\begin{align}
		T_n := (1-\alpha) \cdot M^{\otimes n} = (1-\alpha) \cdot \sum_{k=0}^n t^{n-k} \Lambda_k. \label{eq:optimal_LOCC}
	\end{align}

	Observing that $\Tr[M \rho_0] = 1$ and $\Tr[M \rho_1 ]=t$, we obtain
	\begin{align}
		\alpha_n(T_n) &= \Tr\left[ \left( \mathds{1}- T_n \right) \rho_0^{\otimes n} \right] 
		= 1 - (1-\alpha)\Tr[\rho_0^{\otimes n} M^{\otimes n }]
		= \alpha; \label{eq:alpha_eps}\\
		\beta_n(T_n) &= \Tr\left[ T_n \rho_1^{\otimes n} \right]
		= (1-\alpha) \cdot \Tr\left[ M^{\otimes n} \rho_1^{\otimes n} \right]
		= (1-\alpha) \cdot t^n. \label{eq:beta_eps}
	\end{align}

	Equations \eqref{eq:alpha_eps} and \eqref{eq:beta_eps} imply that $(\alpha, \beta_n(T_n)) \in \mathcal{R}(\rho_0^{\otimes n}, \rho_1^{\otimes n})$.
	Hence, we have $\beta_\alpha^\text{X}(\rho_0^{\otimes n}, \rho_1^{\otimes n}) \leq \beta_n({T_n})= (1-\alpha) \cdot t^n$ as desired. This concludes the proof.
\end{proof}

Combining Proposition~\ref{prop:achievability} and
Proposition~\ref{prop:Exp_LB} in Section~\ref{sec:converse_sym} and noting $\text{LOCC}\subset\text{SEP}\subset\text{PPT}$, we obtain the main result of this section, namely, a sufficient condition of achieving the optimal average error probability using LOCC protocols in symmetric hypothesis testing and the optimal trade-off in asymmetric hypothesis testing. 

\begin{theo}[A sufficient condition for optimal LOCC protocols] \label{theo:sufficient_condition}
	Let $\rho_0$ and $\rho_1$ be arbitrary orthogonal bipartite states.
	If there exists a $t\in(0,1]$ satisfying
	\begin{subequations}
		\begin{align}
			&\left\{ M = t \cdot \Pi_{\rho_1} + \Pi_{\rho_0}, \mathds{1} - M  = (1-t) \cdot \Pi_{\rho_1}   \right\} \in \textnormal{LOCC}, \label{eq:condition_LOCC} 
			\\
			&
	-\rho_1^\Gamma \leq t \cdot \rho_0^\Gamma	\leq \rho_1^\Gamma,
			\label{eq:condition_PPT} 
		\end{align}
	\end{subequations}
	then, for any $n\in\mathds{N}$,
	\begin{align}
		P_\textnormal{e}^\textnormal{PPT} (\rho_0^{\otimes n}, \rho_1^{\otimes n}; p ) = P_\textnormal{e}^\textnormal{SEP} (\rho_0^{\otimes n}, \rho_1^{\otimes n}; p ) = P_\textnormal{e}^\textnormal{LOCC} (\rho_0^{\otimes n}, \rho_1^{\otimes n}; p ) = \min\left\{ (1-p) \cdot t^n, p 
		\right\},
	\end{align}
	and, for any $\alpha \in [0,1]$,
	\begin{align}
		\beta_\alpha^\textnormal{PPT}(\rho_0^{\otimes n}, \rho_1^{\otimes n})  
		= \beta_\alpha^\textnormal{SEP}(\rho_0^{\otimes n}, \rho_1^{\otimes n}) 
		= \beta_\alpha^\textnormal{LOCC}(\rho_0^{\otimes n}, \rho_1^{\otimes n}) 
		= (1-\alpha)\cdot t^n.
	\end{align}
\end{theo}

In Section~\ref{sec:MES}, we will show that $t = \frac{1}{d+1}$ for testing maximally entangled state against its orthogonal complement, and show that $t = \frac{d-1}{d+1}$ for testing extremal Werner states in Section~\ref{sec:Werner}.

\subsection{Testing Maximally Entangled States} \label{sec:MES}

This section aims to test the following hypotheses of maximally entangled state against its orthogonal complement,
\begin{align} \label{eq:MES}
\begin{dcases}
\mathsf{H}_0 : \rho_0^{\otimes n} = \Phi_d^{\otimes n} & \\
\mathsf{H}_1 : \rho_1^{\otimes n } = \left( \Phi_d^\perp \right)^{\otimes n} := \left( \frac{\mathds{1}_{d^2}-\Phi_d}{d^2-1} \right)^{\otimes n} & \\
\end{dcases}, \quad \forall n\in\mathbb{N}.
\end{align}
We establish the optimal average error probability and  the optimal trade-off between the two types 
of errors, type-I and type-II in Theorem~\ref{theo:MES} below.
The Chernoff exponent (Corollary~\ref{coro:MES}) in the 
symmetric setting, the Stein, Hoeffding, and strong converse 
exponents (Corollaries~\ref{coro:MES_Stein}, \ref{coro:MES_Hoeffding}, and \ref{coro:MES_sc}) follow immediately.

\begin{theo} \label{theo:MES}
	Consider the binary hypotheses given in Eq.~\eqref{eq:MES}.
	The following hold for any $n\in\mathds{N}$:
	\begin{align}
		P_\textnormal{e}^\textnormal{PPT} (\rho_0^{\otimes n}, \rho_1^{\otimes n}; p ) = P_\textnormal{e}^\textnormal{SEP} (\rho_0^{\otimes n}, \rho_1^{\otimes n}; p ) = P_\textnormal{e}^\textnormal{LOCC} (\rho_0^{\otimes n}, \rho_1^{\otimes n}; p ) = \min\left\{ (1-p) \left( \frac{1}{d+1} \right)^n, p 
		\right\},
	\end{align}
	and, for any $\alpha\in[0,1]$,
	\begin{align}
	\beta_\alpha^\textnormal{PPT}(\rho_0^{\otimes n}, \rho_1^{\otimes n})  
	= \beta_\alpha^\textnormal{SEP}(\rho_0^{\otimes n}, \rho_1^{\otimes n}) 
	= \beta_\alpha^\textnormal{LOCC}(\rho_0^{\otimes n}, \rho_1^{\otimes n}) 
	= \frac{1-\alpha}{(d+1)^n}.
	\end{align}
\end{theo}

\begin{proof}
		First, we verify condition \eqref{eq:condition_LOCC} in Theorem~\ref{theo:sufficient_condition}.
		For each copy, Alice and Bob measure on the computational basis and compare their measurement outcomes. If they agree, we claim that $\Phi_d^{\otimes n}$ is the true hypothesis.
		We choose such a measurement strategy is because the maximally entangled state is likely to give concordant measurement outcomes due to symmetry.
		Mathematically, this strategy is described by the two-outcome measurement:
	\begin{align}
			\left\{  \sum_{i=0}^{d-1} |ii\rangle\langle ii|, \, \sum_{i\neq j}^{d-1} |ij\rangle\langle ij|
			\right\}.
	\end{align}
	Since the states in the hypotheses are both $U\otimes U^*$-invariant (where `$*$' means the complex conjugate), after the twirling operation, we have the measurement
	\begin{align} \label{eq:Md}
		\left\{M_d =  \Phi_d + \frac{1}{d+1} \left( \mathds{1} - \Phi_d \right), \,
		\mathds{1} - M_d = \frac{d}{d+1}\left( \mathds{1} - \Phi_d \right)
		\right\}.
	\end{align}
	Equivalently, the choice of $M_d$ is implementable by a (one-way) LOCC protocol, and hence \eqref{eq:condition_LOCC} is satisfied with $t = \frac{1}{d+1}$.
	
	Next, 
	we verify that 
	\begin{align}
			\frac{1}{d+1} \left|\Phi_d^\Gamma\right| \leq (\Phi_d^\perp)^\Gamma,
	\end{align}
	and hence 
	the condition \eqref{eq:condition_PPT} in Theorem~\ref{theo:sufficient_condition} with $t= \frac{1}{d+1}$ is fulfilled for the sake of completeness, i.e.~
	\begin{align}
		\frac{1}{d+1} \left|\Phi_d^\Gamma\right| - (\Phi_d^\perp)^\Gamma &= \frac{1}{d^2-1} \left( (d-1) \left|\Phi_d^\Gamma\right| - (\mathds{1}-\Phi_d^\Gamma) \right) \\
		&\leq \frac{1}{d^2-1} \left( (d-1) \left|\Phi_d^\Gamma\right| - \mathds{1} + \left|\Phi_d^\Gamma\right| \right) \\
		&= \frac{1}{d^2-1} \left( d \left|\Phi_d^\Gamma\right| - \mathds{1}  \right) \\
		&=  0,
	\end{align}
	which concludes the proof.
\end{proof}

\begin{remark}
	Theorem~\ref{theo:MES} shows that our result of exponential lower bound for testing arbitrary multipartite entangled pure state against its orthogonal complement given in Proposition~\ref{prop:Exp_LB} of Section~\ref{sec:converse_sym} is tight for maximally entangled state $\Phi_d$.
\end{remark}

From the definition in Eq.~\eqref{eq:Chernoff_X}, we obtain our main result for the Chernoff exponents:
\begin{coro}[Chernoff exponent] \label{coro:MES}
	Consider the binary hypothesis given in Eq.~\eqref{eq:MES}. For every $0<p<1$, we have 
	\begin{align}
		\mathrm{Chernoff}^\textnormal{PPT} (\rho_0, \rho_1) = \mathrm{Chernoff}^\textnormal{SEP}  (\rho_0, \rho_1) = \mathrm{Chernoff}^\textnormal{LOCC} (\rho_0, \rho_1) = \log (d+1).
	\end{align}
\end{coro}

The optimal trade-off between type-I and type-II errors (Theorem~\ref{theo:MES}) gives the results of the Stein exponent (Corollary~\ref{coro:MES_Stein}), the Hoeffding exponent (Corollary~\ref{coro:MES_Hoeffding}),  and strong converse exponent (Corollary~\ref{coro:MES_sc}) as outlined below.

\begin{coro}[Stein exponent] \label{coro:MES_Stein}
	The following Stein exponents hold.
	\begin{enumerate}[(a)]
		\item\label{MES_Stein:a} Consider the binary hypothesis testing: $\rho_0 = \Phi_d$ and $\rho_1 = \Phi_d^{\perp}$. Then,
		\begin{align}
			\mathrm{Stein}^\textnormal{PPT}(\rho_0, \rho_1, \eps) = \mathrm{Stein}^\textnormal{SEP}(\rho_0, \rho_1, \eps) = \mathrm{Stein}^\textnormal{LOCC}(\rho_0, \rho_1,\eps) = \log (d+1), \quad \forall \eps \in [0,1).
		\end{align}
		
		\item\label{MES_Stein:b} Consider the binary hypothesis testing: $\rho_0 = \Phi_d^{\perp}$ and $\rho_1 = \Phi_d$. Then,
		\begin{align}
			\mathrm{Stein}^\textnormal{PPT}(\rho_0, \rho_1, \eps) = \mathrm{Stein}^\textnormal{SEP}(\rho_0, \rho_1, \eps) = \mathrm{Stein}^\textnormal{LOCC}(\rho_0, \rho_1,\eps) = \infty, \quad \forall \eps \in [0,1).
		\end{align}
	\end{enumerate}	
\end{coro}
\begin{proof}
	\ref{MES_Stein:a}
	By Theorem~\ref{theo:MES}, we have, for every $\eps\in[0,1)$,
	\begin{align}
		\mathrm{Stein}^\text{LOCC} (\Phi_d^{\otimes n}, (\Phi_d^{\perp})^{\otimes n}, \eps)
		= \lim_{n\to \infty} \beta_\eps^\text{X}(n) = \lim_{n\to \infty} \log (d+1) - \frac1n \log (1-\eps) = \log (d+1)
	\end{align}
	as desired.
	
	\ref{MES_Stein:b}
	For every $\eps \in (0,1)$,
	we choose the test $T_n := \mathds{1} -  M_d^{\otimes n}$.
	From Eqs.~\eqref{eq:alpha_eps} and \eqref{eq:beta_eps}, we have
	\begin{align}
		\alpha_n(T_n) = \frac{1}{(d+1)^n}; \quad \beta_n(T_n)  = 0.
	\end{align}
	Hence, it follows that $ \lim_{n\to 0} \alpha_n(T_n) \leq \eps$, and
	\begin{align}
		\mathrm{Stein}^\text{X} (\rho_0, \rho_1, \eps) \geq \lim_{n\to \infty} -\frac1n \log \beta_n(T_n) = \infty.
	\end{align}
	This completes the proof.
\end{proof}


\begin{coro}[Hoeffding exponent] \label{coro:MES_Hoeffding}
	The following Hoeffding exponents hold.
	\begin{enumerate}[(a)]
		\item\label{MES_Hoeffding:a} Consider the binary hypothesis testing: $\rho_0 = \Phi_d$ and $\rho_1 = \Phi_d^{\perp}$. Then,
		\begin{align}
			\mathrm{Hoeffding}^\textnormal{PPT}(\rho_0, \rho_1, r) = \mathrm{Hoeffding}^\textnormal{SEP}(\rho_0, \rho_1, r) = \mathrm{Hoeffding}^\textnormal{LOCC}(\rho_0, \rho_1, r) = \infty, \quad \forall r \leq \log (d+1).
		\end{align}
		
		\item\label{MES_Hoeffding:b} Consider the binary hypothesis testing: $\rho_0 = \Phi_d^{\perp}$ and $\rho_1 = \Phi_d$. Then,
		\begin{align}
			\mathrm{Hoeffding}^\textnormal{PPT}(\rho_0, \rho_1, r) = \mathrm{Hoeffding}^\textnormal{SEP}(\rho_0, \rho_1, r) = \mathrm{Hoeffding}^\textnormal{LOCC}(\rho_0, \rho_1, r) = \log (d+1), \quad \forall r> 0.
		\end{align}
	\end{enumerate}
\end{coro}
\begin{proof}
	\ref{MES_Hoeffding:a}		
	From Theorem~\ref{theo:MES}, we know that if the type-II error is allowed to decay exponentially at the rate of $\log (d+1)$, then the type-I error is zero for all $n$. From the definition of the Hoeffding exponent given in Eq.~\eqref{eq:Hoeffding}, the type-I error is always zero if the type-II error decays slower, which shows that the Hoeffding exponent is infinite.
	
	\ref{MES_Hoeffding:b}
	Fix an arbitrary $n$.
	From Theorem~\ref{theo:MES}, Eqs.~\eqref{eq:alpha_eps} and \eqref{eq:beta_eps}, choosing the test as $\mathds{1} - (1-\eps) M_d^{\otimes n}$,
	it follows that 
	\begin{align}
		\alpha_n(T_n) = \frac{1-\eps}{(d+1)^n}, \quad \beta_n(T_n) = \eps.
	\end{align}
	Note that the test is optimal for every $\eps \in [0,1)$.
	Now,  letting $\eps = \exp\{-nr\}$ for any $r> 0$, we have $\alpha_n(T_n) = (1-\exp\{-nr\})$. Since this holds for every $n$, then
	\begin{align}
	 \mathrm{Hoeffding}^\text{LOCC}  (\rho_0, \rho_1, r) = \lim_{n\to\infty} -\frac1n \log \alpha_n(T_n) 
		= \log (d+1) - \lim_{n\to \infty} \frac1n \log (1-\exp\{-nr\})
		= \log(d+1),
	\end{align}
	thereby completing the proof.
\end{proof}

\begin{coro}[Strong converse exponent] \label{coro:MES_sc}
	The following strong converse exponents hold.
	\begin{enumerate}[(a)]
		\item\label{MES_sc:a} Consider the binary hypothesis testing: $\rho_0 = \Phi_d$ and $\rho_1 = \Phi_d^{\perp}$. Then,
		\begin{align}
			\mathrm{SC}^\textnormal{PPT}(\rho_0, \rho_1, r) = \mathrm{SC}^\textnormal{SEP}(\rho_0, \rho_1, r) = \mathrm{SC}^\textnormal{LOCC}(\rho_0, \rho_1, r) = r - \log(d+1), \quad \forall r > \log (d+1).
		\end{align}
		
		\item\label{MES_sc:b} Consider the binary hypothesis testing: $\rho_0 = \Phi_d^{\perp}$ and $\rho_1 = \Phi_d$. Then,
		\begin{align}
		\mathrm{SC}^\textnormal{PPT}(\rho_0, \rho_1, r) = \mathrm{SC}^\textnormal{SEP}(\rho_0, \rho_1, r) = \mathrm{SC}^\textnormal{LOCC}(\rho_0, \rho_1, r) = 0, \quad \forall r\geq 0.
		\end{align}
	\end{enumerate}
\end{coro}
\begin{proof}
	\ref{MES_sc:a}
	By denoting $\alpha_\beta^\text{X} := \inf \left\{ \alpha: (\alpha,\beta) \in \mathcal{R}^\text{X}(\rho_0^{\otimes n}, \rho_1^{\otimes n}) \right\}$, Theorem~\ref{theo:MES} implies that
	\begin{align}
		\beta = \frac{1-\alpha_\beta^\text{X} }{(d+1)^n}.
	\end{align}
	Letting $\beta = \exp\{-nr\}$, we have the desired strong converse exponent.
	
	\ref{MES_sc:b}	 
	By Corollary~\ref{coro:MES_Hoeffding}-\ref{MES_Hoeffding:b}, the type-I error will exponentially decay when the type-II error exponentially decays at any finite rate $r$. Therefore, the strong converse exponent is zero for all finite rate $r$ since $\lim_{n\to \infty} \frac1n \log [1- \frac{1}{(d+1)^n}  ] = 0$.
\end{proof}

We remark that analysis of this section directly applies to the scenario of testing pure states with uniform non-zero Schmidt coefficients. We defer the results to Appendix~\ref{sec:MES_high} (see also Table~\ref{table:summary}).

\subsection{Testing Completely (Anti-)Symmetric Werner States} \label{sec:Werner}

In this section, we consider testing 
completely anti-symmetric Werner state $\sigma_{\text{a},d}$ and
completely symmetric Werner state $\sigma_{\text{s},d}$ \cite{DLT02, Wer89}:
\begin{align} \label{eq:Werner}
	\begin{dcases}
		\mathsf{H}_0 : \rho_0^{\otimes n} = \sigma_{\text{a},d}^{\otimes n}
		:= \left( \frac{2}{d(d-1)} \Pi_{\text{a},d} \right)^{\otimes n}\\
		\mathsf{H}_1 : \rho_1^{\otimes n} = \sigma_{\text{s},d}^{\otimes n} := \left(\frac{2}{d(d+1)} \Pi_{\text{s},d}\right)^{\otimes n}\\
	\end{dcases}, \quad \forall n\in\mathbb{N}.
\end{align}
Here, $\Pi_{\text{a},d} := \frac12 (\mathds{1}_{d^2} - F_d )$ and $\Pi_{\text{s},d} := \frac12 (\mathds{1}_{d^2} + F_d )$ and  are projections onto the anti-symmetric and symmetric subspaces, respectively, and $F_d := \sum_{i,j=0}^{d-1} |i\rangle \langle j| \otimes |j\rangle \langle i|$ is the swap operator between the two $d$-dimensional subsystems
We obtain the following optimal LOCC protocols.

\begin{theo} \label{theo:Werner}
	Consider the binary hypotheses given in Eq.~\eqref{eq:Werner}.
	The following hold for any $n\in\mathds{N}$:
	\begin{align}
		P_\textnormal{e}^\textnormal{PPT} (\rho_0^{\otimes n}, \rho_1^{\otimes n}; p ) = P_\textnormal{e}^\textnormal{SEP} (\rho_0^{\otimes n}, \rho_1^{\otimes n}; p ) = P_\textnormal{e}^\textnormal{LOCC} (\rho_0^{\otimes n}, \rho_1^{\otimes n}; p ) = \min\left\{ (1-p) \left( \frac{d-1}{d+1} \right)^n, p 
		\right\},
	\end{align}
	and, for any $\alpha\in[0,1]$,
	\begin{align}
		\beta_\alpha^\textnormal{PPT}(\rho_0^{\otimes n}, \rho_1^{\otimes n})  
		= \beta_\alpha^\textnormal{SEP}(\rho_0^{\otimes n}, \rho_1^{\otimes n}) 
		= \beta_\alpha^\textnormal{LOCC}(\rho_0^{\otimes n}, \rho_1^{\otimes n}) 
		= (1-\alpha)\cdot \left( \frac{d-1}{d+1} \right)^n.
	\end{align}
\end{theo}
\begin{proof}
	For verifying the first condition \eqref{eq:condition_LOCC} in Theorem~\ref{theo:sufficient_condition}, we consider the following (one-way) LOCC protocol for each copy:
	\begin{align}
	\left\{  \sum_{i\neq j}^{d-1} |ij\rangle\langle ij|, \, \sum_{i=0}^{d-1} |ii\rangle\langle ii|
	\right\}.
	\end{align}
	Noting that the states $\sigma_{\text{a},d}$ and $\sigma_{\text{s},d}$ are both $U\otimes U$-invariant, after twirling, we obtain:
	\begin{align}
		\left\{ M_d :=  \frac{d-1}{d+1} \Pi_{\text{s},d} + \Pi_{\text{a},d}, \, \mathds{1} - M_d = \frac{2}{d+1} \Pi_{\text{s},d}
		\right\}.
	\end{align}
	This shows that \eqref{eq:condition_LOCC} holds for $t = \frac{d-1}{d+1}$.
	
	To verify the second condition \eqref{eq:condition_PPT} with $t = \frac{d-1}{d+1}$, we calculate
	\begin{align}
	\frac{d-1}{d+1} \left|\sigma_{\text{a},d}^\Gamma\right| - \sigma_{\text{s},d}^\Gamma 
	&= \frac{1}{d(d+1)} \left( \left|\mathds{1}_{d^2} - F_d^\Gamma\right| - \mathds{1}_{d^2} - F_d^\Gamma \right) \\
	&= \frac{1}{d(d+1)} \left( \left|\mathds{1}_{d^2} - d \Phi_d \right| - \mathds{1}_{d^2} - d \Phi_d \right) \\
	&= - \frac{2}{d(d+1)} \Phi_d \\
	&\leq 0.
	\end{align}
This guarantees that \eqref{eq:condition_PPT} must be fulfilled.
\end{proof}

Then, the associated Stein, Hoeffding, and the strong converse exponents can be established similarly as in Corollaries~\ref{coro:MES_Stein}, \ref{coro:MES_Hoeffding}, and \ref{coro:MES_sc}. We lists the results below without duplicating the proof (see also Table~\ref{table:summary} for a summary).

\begin{coro}[Stein exponent] \label{coro:Werner_Stein}
	The following Stein exponents hold.
	\begin{enumerate}[(a)]
		\item Consider the binary hypothesis testing: $\rho_0 = \sigma_{\textnormal{a},d}$ and $\rho_1 = \sigma_{\textnormal{s},d}$. Then,
		\begin{align}
			\mathrm{Stein}^\textnormal{PPT}(\rho_0, \rho_1, \eps) = \mathrm{Stein}^\textnormal{SEP}(\rho_0, \rho_1, \eps) = \mathrm{Stein}^\textnormal{LOCC}(\rho_0, \rho_1,\eps) = \log \frac{d+1}{d-1}, \quad \forall \eps \in [0,1).
		\end{align}
		
		\item Consider the binary hypothesis testing: $\rho_0 = \sigma_{\textnormal{s},d}$ and $\rho_1 = \sigma_{\textnormal{a},d}$. Then,
		\begin{align}
			\mathrm{Stein}^\textnormal{PPT}(\rho_0, \rho_1, \eps) = \mathrm{Stein}^\textnormal{SEP}(\rho_0, \rho_1, \eps) = \mathrm{Stein}^\textnormal{LOCC}(\rho_0, \rho_1,\eps) = \infty, \quad \forall \eps \in [0,1).
		\end{align}
	\end{enumerate}	
\end{coro}

\begin{coro}[Hoeffding exponent] \label{coro:Werner_Hoeffding}
	The following Hoeffding exponents hold.
	\begin{enumerate}[(a)]
		\item\label Consider the binary hypothesis testing: $\rho_0 = \sigma_{\textnormal{a},d}$ and $\rho_1 = \sigma_{\textnormal{s},d}$. Then,
		\begin{align}
			\mathrm{Hoeffding}^\textnormal{PPT}(\rho_0, \rho_1, r) = \mathrm{Hoeffding}^\textnormal{SEP}(\rho_0, \rho_1, r) = \mathrm{Hoeffding}^\textnormal{LOCC}(\rho_0, \rho_1, r) = \infty, \quad \forall r \leq \log \frac{d+1}{d-1}.
		\end{align}
		
		\item\label Consider the binary hypothesis testing: $\rho_0 = \sigma_{\textnormal{s},d}$ and $\rho_1 = \sigma_{\textnormal{a},d}$. Then,
		\begin{align}
			\mathrm{Hoeffding}^\textnormal{PPT}(\rho_0, \rho_1, r) = \mathrm{Hoeffding}^\textnormal{SEP}(\rho_0, \rho_1, r) = \mathrm{Hoeffding}^\textnormal{LOCC}(\rho_0, \rho_1,r) = \log \frac{d+1}{d-1}, \quad \forall r> 0.
		\end{align}
	\end{enumerate}
\end{coro}

\begin{coro}[Strong converse exponent] \label{coro:Werner_sc}
	The following strong converse exponents hold.
	\begin{enumerate}[(a)]
		\item Consider the binary hypothesis testing: $\rho_0 = \sigma_{\textnormal{a},d}$ and $\rho_1 = \sigma_{\textnormal{s},d}$. Then,
		\begin{align}
			\mathrm{SC}^\textnormal{PPT}(\rho_0, \rho_1, r) = \mathrm{SC}^\textnormal{SEP}(\rho_0, \rho_1, r) = \mathrm{SC}^\textnormal{LOCC}(\rho_0, \rho_1, r) = r - \log\frac{d+1}{d-1}, \quad \forall r > \log \frac{d+1}{d-1}.
		\end{align}
		
		\item  Consider the binary hypothesis testing: $\rho_0 = \sigma_{\textnormal{s},d}$ and $\rho_1 = \sigma_{\textnormal{a},d}$. Then,
		\begin{align}
			\mathrm{SC}^\textnormal{PPT}(\rho_0, \rho_1, r) = \mathrm{SC}^\textnormal{SEP}(\rho_0, \rho_1, r) = \mathrm{SC}^\textnormal{LOCC}(\rho_0, \rho_1, r) = 0, \quad \forall r\geq 0.
		\end{align}
	\end{enumerate}
\end{coro}

\section{Infinite asymptotic separation between SEP POVMs and PPT POVMs}
\label{sec:separation}
In this section, we prove that there is an \emph{infinite {asymptotic} separation} between the optimal 
average error probabilities using SEP POVMs and PPT POVMs.

\begin{defn}[Unextendible product basis \cite{BDM+99, DMS+03}] \label{defn:UPB}
	Consider two Hilbert spaces, $\mathcal{H}_A$ and $\mathcal{H}_B$. The set $S = \{ |\alpha_i\rangle \otimes |\beta_i\rangle : 1\leq i \leq N\} \subseteq \mathcal{H}_A\otimes \mathcal{H}_B$ is called an \emph{unextendible product basis}\footnote{In this paper we only considered unextendible product bases on a bipartite Hilbert spaces. However, the Definition~\ref{defn:UPB} naturally extends to the multipartite scenario.} (UPB) if it satisfies the following:
	\begin{align}
		\begin{dcases}
			&|\alpha_i\rangle \otimes |\beta_i\rangle \perp |\alpha_j\rangle \otimes |\beta_j\rangle, \quad \forall i\neq j \\
			&\left\{
			|\phi\rangle \otimes |\psi\rangle \in \mathcal{H}_A\otimes \mathcal{H}_B: \forall i \in\{1,\ldots, N\},	|\phi\rangle \otimes |\psi\rangle \perp |\alpha_i\rangle \otimes |\beta_i\rangle
			\right\} = {\{0\}} \\
		\end{dcases}.
	\end{align}
\end{defn}
An example of UBP was given in \cite[(18)]{BDM+99}, which consists of five states on $\mathds{C}^3\otimes \mathds{C}^3$:
\begin{align*}
	&|\psi_0\rangle = \frac{1}{\sqrt{2}}|0\rangle \otimes (|0\rangle - |1\rangle ), \quad 
	|\psi_2\rangle = \frac{1}{\sqrt{2}}|2\rangle \otimes (|1\rangle - |2\rangle ), \\
	&|\psi_1\rangle = \frac{1}{\sqrt{2}}  (|0\rangle - |1\rangle )\otimes|2\rangle, \quad 
	|\psi_3\rangle = \frac{1}{\sqrt{2}} (|1\rangle - |2\rangle )\otimes  |0\rangle, \\
	&\quad |\psi_4\rangle = \frac{1}{3}(|0\rangle +|1\rangle + |2\rangle ) \otimes  (|0\rangle +|1\rangle + |2\rangle ).
\end{align*}

For a UPB
$S = \{ |\alpha_i\rangle \otimes |\beta_i\rangle : 1\leq i \leq N\} \subseteq \mathcal{H}_A\otimes \mathcal{H}_B$, we consider the following binary hypotheses:
\begin{align} \label{eq:example}
	\begin{dcases}
		\mathsf{H}_0 : \rho_0^{\otimes n} := \left[ \frac1N \left( \sum_{i=1}^N|\alpha_i \beta_i\rangle \langle \alpha_i \beta_i | \right) \right]^{\otimes n} \\
		\mathsf{H}_1:\rho_1^{\otimes n}
	\end{dcases}, n  \in\mathbb{N}.
\end{align}
Here,  ${\rho_1}$ is a state in $\mathcal{H}_A\otimes \mathcal{H}_B$ that is orthogonal to ${\rho_0}$.
We show below that the optimal exponential decay rate of using SEP POVMs is strictly worse than that of using PPT POVMs.

\begin{theo}[Separation] \label{theo:separation}
	There is an infinite {asymptotic} separation between SEP POVMs and PPT POVMs. That is, for testing the states given in Eq.~\eqref{eq:example}, there exists a $\mu>0$ such that  {for every $0<p<1$,}
	\begin{align}
		\begin{cases}
			{P_\textnormal{e}^\textnormal{PPT}(\rho_0^{\otimes n}, \rho_1^{\otimes n}; p)} = 0 \\
			{P_\textnormal{e}^\textnormal{SEP}(\rho_0^{\otimes n}, \rho_1^{\otimes n}; p)} \geq \frac{\mu^n}{2}	\\	
		\end{cases}, \quad \forall n\in \mathbb{N}.
	\end{align}
Consequently, we have that $\mathrm{Chernoff}^\textnormal{PPT}(\rho_0, \rho_1) = \infty$, while $\mathrm{Chernoff}^\textnormal{SEP}(\rho_0, \rho_1) < \infty$.

\end{theo}

The key ingredient to establish Theorem~\ref{theo:separation} is to introduce a novel 
quantity that characterizes the ``richness'' of a product basis:

\begin{defn}[Unextendibility] \label{defn:delta}
	Given a product basis $S = \left\{|\alpha_i \rangle \otimes |\beta_i\rangle : 1\leq i \leq N \right\} \subseteq \mathcal{H}_A \otimes \mathcal{H}_B$, we define
	the unextendibility of $S$ as
\begin{align} \label{eq:delta}
	\delta_S := \min_{\rho_A, \rho_B} \max_{1\leq i\leq N} \langle \alpha_i |\rho_A | \alpha_i\rangle \langle \beta_i | \rho_B | \beta_i \rangle ,
\end{align}
where $\rho_A$  and $\rho_B$ range over all density matrices on $\mathcal{H}_A$ and $\mathcal{H}_B$, respectively.
\end{defn}

Note that the minimum in Eq.~\eqref{eq:delta} is always attained for some $\rho_A$ and $\rho_B$,
since the objective function on the right-hand side is continuous in $\rho_A\otimes \rho_B$ 
and the set of product states $\rho_A\otimes \rho_B$ is compact.

For an {unextendible product basis} $S$ \cite{DMS+03}, one immediately has
$\delta_S > 0$.

\begin{prop}[Multiplicativity] \label{prop:multiplicative}
	The quantity $\delta_S$ is multiplicative.
	That is, 
	\begin{align}
		\delta_{S_1\otimes S_2} = \delta_{S_1} \delta_{S_2}
	\end{align}
	for any two product bases $S_1$ and $S_2$.
\end{prop}

\begin{remark}
	Note that a product basis $S$ is unextendible if and only if $\delta_S > 0$.
	Hence, the quantity $\delta_S$ indicates how far a UPB is from an extendible product basis.
	Then Proposition~\ref{prop:multiplicative} implies that the tensor product of any 
	two UPBs $S_1$ and $S_2$ enjoys the property that $\delta_{S_1\otimes S_2} = \delta_{S1} \delta_{S_2} > 0$. Therefore, this gives a quantitative characterization of the well-known fact that the tensor 
	product of UPBs is also a UPB \cite{DMS+03}.
\end{remark}

\begin{proof}
Let two product bases be given as
	\begin{align}
		S_1 &= \left\{|\alpha_i \rangle \otimes |\beta_i\rangle : 1\leq i \leq N_1 \right\} \subseteq \mathcal{H}_{A_1} \otimes \mathcal{H}_{B_1}, \\	
		S_2 &= \left\{|\psi_j \rangle \otimes |\phi_j\rangle : 1\leq j \leq N_2 \right\} \subseteq \mathcal{H}_{A_2} \otimes \mathcal{H}_{B_2}.
	\end{align}
	From the definition given in Eq.~\eqref{eq:delta}, one can directly verify that
	\begin{align}
		\delta_{S_1\otimes S_2} \leq \delta_{S_1} \delta_{S_2}.
	\end{align}
	It remains to prove the direction $\geq$.
	
	Let $\rho_{A_1 A_2}$ on $\mathcal{H}_{A_1} \otimes \mathcal{H}_{A_2}$ and $\rho_{B_1 B_2} $ on $\mathcal{H}_{B_1} \otimes \mathcal{H}_{B_2}$ be the states that attain the minimization in  $\delta_{S_1\otimes S_2}$, i.e.,
	\begin{align} \label{eq:S1S2}
		\delta_{S_1\otimes S_2} = \max_{ \substack{ 1\leq i \leq N_1 \\  1\leq j \leq N_2}  }
		\langle \alpha_{i} \psi_j | \rho_{A_1 A_2} | \alpha_{i} \psi_j \rangle 
		\langle \beta_{i} \phi_j | \rho_{B_1 B_2} | \beta_{i} \phi_j \rangle.
	\end{align}
	Further, for any $1\leq i \leq N_1$, we let
	\begin{align}
		\langle \alpha_i | \rho_{A_1 A_2} | \alpha_i \rangle &= \lambda_i \sigma_{i,A_2}, \\
		\langle \beta_i  | \rho_{B_1 B_2} | \beta_i \rangle &= \mu_i \sigma_{i,B_2},
	\end{align}
	where
	\begin{align}
		\lambda_i &:= \Tr\left[ \langle \alpha_i | \rho_{A_1 A_2} | \alpha_i \rangle  \right] = \langle \alpha_i | \rho_{A_1} | \alpha_i \rangle  \geq 0, \\
		\mu_i &:= \Tr\left[ \langle \beta_i | \rho_{B_1 B_2} | \beta_i \rangle  \right] = \langle \beta_i | \rho_{B_1} | \beta_i \rangle  \geq 0.
	\end{align}
	By the definition of $\delta_1$, there always exists some $i_0$ such that
	\begin{align} \label{eq:S1}
		\lambda_{i_0} \mu_{i_0} \geq \delta_{S_1}.
	\end{align}
	
	
	By the choice of $i_0$ for any $1\leq j \leq N_2$, we have
	\begin{align}
		\langle \alpha_{i_0} \psi_j | \rho_{A_1 A_2} | \alpha_{i_0} \psi_j \rangle 
		\langle \beta_{i_0} \phi_j | \rho_{B_1 B_2} | \beta_{i_0} \phi_j \rangle
		= \lambda_{i_0} \mu_{i_0} \langle \psi_j |\sigma_{i_0, A_2} | \psi_j \rangle
		\langle \phi_j | \sigma_{i_0, B_2} | \phi_j \rangle.
	\end{align}
	By the definition of $\delta_{S_2}$, there exists some $j_0$ such that
	\begin{align} \label{eq:S2}
		\langle \psi_{j_0} | \sigma_{i_0, A_2} | \psi_{j_0} \rangle
		\langle \phi_{j_0} | \sigma_{i_0, B_2} | \phi_{j_0} \rangle \geq \delta_{S_2}.
	\end{align}
	Combining Eqs.~\eqref{eq:S1S2}, \eqref{eq:S1}, and \eqref{eq:S2} gives
	\begin{align}
		\delta_{S_1\otimes S_2} \geq \delta_{S_1}\delta_{S_2},
	\end{align}
	which completes the proof.
\end{proof}
Now, we are ready to prove the main result, expressed as Theorem~\ref{theo:separation} in this section.
\begin{proof}[Proof of Theorem~\ref{theo:separation}]
	Consider the binary hypotheses given in Eq.~\eqref{eq:example} and let
	\begin{align}
        P = \sum_{i=1}^N |\alpha_i \beta_i\rangle \langle \alpha_i \beta_i |, \quad
		\rho = \frac{1}{N} P.
	\end{align}
	By construction, both $P$ and $\mathds{1}-P$ are PPT operators \cite{DMS+03}.
	For this reason, the pair of states $(\rho,\sigma)$ can be perfectly distinguished by PPT POVMs. Hence, it remains to show an exponential lower bound to the error probability for distinguishing $(\rho_0^{\otimes n}, \rho_1^{\otimes n})_{n \in \mathbb{N}}$ using SEP POVMs.
	
	Consider uniform prior. Then, we have
	\begin{align}
		P_\text{e}^\text{SEP} \left( \rho_0^{\otimes n} , \rho_1^{\otimes n}; \tfrac12 \right)
		&= 1 - \sup_{T, (\mathbb{1}-T) \in \text{SEP} } \left\{ \frac12 \Tr\left[\rho_0^{\otimes n} (\mathbb{1}-T)\right] + \frac12\Tr\left[ \rho_1^{\otimes n} T  \right] \right\}.
	\end{align}
The above equation can be written as a primal problem of a \emph{semi-definite program} as follows:
	\begin{align}
		\text{maximize}: \quad&\frac12 \Tr\left[\rho_0^{\otimes n} \Pi_0 \right] + \frac12\Tr\left[ \rho_1^{\otimes n}  \Pi_1 \right];  \\
		\text{subject to}: \quad &\Pi_0 + \Pi_1 = \mathbb{1}, \\
		&\Pi_i \in \text{SEP} \quad \forall i\in \{0,1\}.
	\end{align}	
	To derive its dual problem, we introduce the \emph{dual cone} to SEP as follows:
	\begin{align}
		\text{SEP}^*(A^n: B^n) := \left\{ H: H^\dagger = H, \; \Tr\left[ \Pi \cdot H \right] \geq 0, \quad \forall \, \Pi \in \text{SEP}
		\right\}.
	\end{align}
	Such a dual cone is also known as the set of \emph{block-positive operators} \cite[Section 2]{BCJ+15}, i.e.,
	\begin{align}
		\text{SEP}^*(A^n: B^n) = \left\{ H: H^\dagger = H, \;  \Tr_{B^n}\left[  H \cdot (\mathbb{1}_{A^n} \otimes |y\rangle\langle y|) \right] 
	\geq 0, \quad \forall |y\rangle  \in \mathcal{H}_{B}^{\otimes n}
	\right\}.
	\end{align}
	Then the associated dual problem is:
	\begin{align}
		\text{minimize}: \quad&\Tr[H]; \\
		\text{subject to}: \quad &H - \frac12 \rho_0^{\otimes n} \in \text{SEP}^*(A^n: B^n), \\
		&H - \frac12 \rho_1^{\otimes n} \in \text{SEP}^*(A^n: B^n),\\
		&H = H^\dagger.
	\end{align}

	By the weak duality, we have
	\begin{align}
		P_\text{e}^\text{SEP} \left( \rho_0^{\otimes n} , \rho_1^{\otimes n}; \tfrac12 \right)&\geq 1 - \inf_{H^\dagger = H} \left\{ \Tr[H]: H - \frac12 \rho_0^{\otimes n}, H - \frac12 \rho_1^{\otimes n} \in \text{SEP}^*(A^n: B^n)
		\right\}. \label{eq:feasible_SEP}
	\end{align}

	Now, we choose a Hermitian operator $H$ as
	\begin{align}
		H &= \frac12 \rho_0^{\otimes n} + \left( \frac12 - \frac{\mu^n}{2} \right) \rho_1^{\otimes n}, \\
		\mu &:= \frac{\delta_S}{N} \in (0,1),
	\end{align}
	where $\delta_S$ was introduced in Definition~\ref{defn:delta}, and it is clear that $\delta_S \leq 1$.
	We aim to show that the operator $H$ is a feasible solution to Eq.~\eqref{eq:feasible_SEP} to complete the proof.
	First, it is not hard to see that
	\begin{align}
		H - \frac12 \rho_0^{\otimes n}=  \left( \frac12 - \frac{\mu^n}{2} \right) \rho_1^{\otimes n} \in \text{SEP}^*(A^n: B^n),
	\end{align}
	since positivity implies block-positivity.
	Second, we have
	\begin{align}
		H - \frac12 \rho_1^{\otimes n} &=  \frac12 \left( \rho_0^{\otimes n} -\mu^n\rho_1^{\otimes n } \right) \\
		&\geq \frac12 \left( \rho_0^{\otimes n} -\mu^n \mathbb{1}^{\otimes n } \right) \\
		&= \frac{1}{2N^n} \left( P^{\otimes n} - \delta_S^n \mathbb{1}^{\otimes n } \right).
	\end{align}

	To show that the quantity $P^{\otimes n} - \delta_S^n \mathbb{1}^{\otimes n }$ is block-positive, we will invoke the definition of $\delta_S$ given in Definition~\ref{defn:delta} and the multiplicativity of $\delta_S$ established in Proposition~\ref{prop:multiplicative}.
	Note that for all density matrices $\omega_A$ on $\mathcal{H}_A$ and $\omega_B$ on $\mathcal{H}_B$, we have
	\begin{align}
		\Tr\left[ P \cdot \omega_A\otimes \omega_B  \right] \geq \delta_S > 0.
	\end{align}
	Hence, for any $\omega_{A_1 A_2 \ldots A_n} \otimes \omega_{B_1	 B_2 \ldots B_n}$, 
	\begin{align}
		\Tr\left[ P^{\otimes n} \cdot  \omega_{A_1 A_2 \ldots A_n} \otimes \omega_{B_1	 B_2 \ldots B_n} 	\right] \geq \delta_{S^{\otimes n}} = \delta_S^n >0
	\end{align}
	In other words, $P^{\otimes n} - \delta_S^n \mathbb{1}^{\otimes n}$ is block-positive, which proves our claim.
\end{proof}

\section{Discussions and Conclusions} \label{sec:conclusions}
We studied the hypothesis testing between an entangled pure state against its orthogonal complement in the many-copy scenario. Our two principal motivations for this work are as follows: First, quantum entanglement is a valuable resource in quantum computation and other important quantum information-theoretic protocols. How to distinguish a state which possesses entangled bits is of fundamental significance. Second, whether the Chernoff exponent for pairs of orthogonal states with entanglement under LOCC is faithful (i.e., finite) is a long-term open problem. This problem is challenging because there is no known simple mathematical structure of LOCC. Moreover, a general state could be entangled in a multipartite quantum system. Without any simple mathematical expression, hypothesis testing using LOCC is more challenging.

In this paper, we show that the optimal average error probability for testing an arbitrary multipartite entangled pure state against its orthogonal complement decays exponentially in the number of copies, which in turn implies that the associated Chernoff exponent is faithful. In the special case of a maximally entangled state, we explicitly derive its optimal average error probability in the symmetric setting and show that a single-copy measurement performs as well as the best LOCC protocol. This finding directly leads to an explicit expression of the Chernoff exponent.

In the asymmetric setting, we obtain an optimal trade-off between the type-I error $\alpha_n$ and the type-II error $\beta_n$. This allows us to fully characterize their asymptotic properties.
The associated Stein's exponent, namely the optimal exponential rate of $\beta_n$ when $\alpha_n$ is at most a constant, is proved.
When $\beta_n$ decays at a rate below or above the Stein exponent, we also establish the optimal exponential rate of $\alpha_n$ and $1-\alpha_n$, respectively.
Our results show that the negative logarithm of the considered error does not diverge, which then guarantees that all the definitions of the four exponents are faithful.
In other words, the asymptotic behavior of the errors satisfies the so-called \emph{large deviation principle} \cite{DZ98}. It is worth mentioning that there is a distinctive difference between the second-order asymptotics under LOCC and those under global measurements. Indeed, this signifies that LOCC has an exceptional structure. To be more specific, we recall the second-order expansion of the optimal exponential rate of type-II error with type-I error no larger than a constant $\eps$ under global measurements \cite{TH13, Li14}, i.e.,~
	\begin{align} \label{eq:second-order}
		D(\rho_0\| \rho_1) + \sqrt{\frac{V(\rho_0\|\rho_1)}{n} } Q^{-1}(\eps) + O\left(\frac{\log {n}}{n}\right),
	\end{align}
	where $D$ and $V$ are the quantum relative entropy and the relative entropy variance, respectively, and $Q$ is the cumulative distribution function of the standard normal distribution.
	As described in \cite[Section 5]{Li14}, $V(\rho_0\|\rho_1) = 0 $ implies $D(\rho_0\|\rho_1) = 0$ for any pair of quantum states $\rho_0,\rho_1$. In other words, the first-order term in Eq.~\eqref{eq:second-order} disappears whenever the second-order term vanishes.
	On the other hand, the obtained result for LOCC in Eq.~\eqref{eq:trade-off0}, shows that the optimal exponential rate of type-II error is given by
	\begin{align} \label{eq:ours0}
		-\frac1n \log \beta_n = \log (d+1) - \frac{\log (1-\eps)}{n}.
	\end{align}
	While the second-order term is missing, the first-order term $\log(d+1)$ is strictly positive. Indeed, this result implies that such asymptotic expansion has a different second-order term from that in Eq.~\eqref{eq:second-order}.
In previous work by Hayashi and Owari \cite{OH14, OH15, HO17},we note that when distinguishing a bipartite pure state and a completely mixed state under LOCC, the asymptotic bounds admit a similar second-order expansion as in Eq.~\eqref{eq:second-order} unless the pure state is maximally entangled.

Finally, we establish an infinite {asymptotic} separation between the SEP and PPT operations in the many-copy scenario. Our result shows that indeed there is a gap between the SEP and PPT operations no matter how many copies of states are provided.
Our technique is a multiplicativity property---a quantitative characterization of the tensor product of unextendible product bases. On the other hand, whether there is an infinite separation between the SEP and LOCC operations is a compelling open problem for future work. We believe that our analysis and results might have applications in data hiding or the studies of other important sets of orthogonal states.

\section*{Acknowledgements}
We sincerely thank anonymous reviewers for their insightful suggestions to improve this paper. In particular, proofs of Theorem~\ref{theo:MES} and Theorem~\ref{theo:Werner} are greatly simplified.
Also, the hypothesis in Lemma~\ref{lemm:PPT} is more  precise.

HC is supported by the Young Scholar Fellowship (Einstein Program) of the Ministry of Science and Technology, Taiwan (R.O.C.) under Grants NSTC 111-2636-E-002-026, NSTC 112-2636-E-002-009, NSTC 112-2119-M-007-006, NSTC 112-2119-M-001-006, NSTC 112-2124-M-002-003,
by the Yushan Young Scholar Program of the Ministry of Education, Taiwan (R.O.C.) under Grants NTU-111V1904-3 and NTU-112V1904-4,
and by the research project ``Pioneering Research in Forefront Quantum Computing, Learning and Engineering'' of National Taiwan University under Grant NTC-CC-112L893405.''
A.~Winter acknowledges financial support by the Spanish MINECO (projects
FIS2016-86681-P and PID2019-107609GB-I00) with the support of FEDER
funds, and the Generalitat de Catalunya (project CIRIT 2017-SGR-1127).
N.~Yu is supported by ARC Discovery Early Career Researcher Award DE180100156 and ARC Discovery Project DP210102449.

\appendix

\section{Testing Pure States with Uniform Non-Zero Schmidt Coefficients} \label{sec:MES_high}

In the following, we test pure entangled states with equal positive Schmidt coefficients, i.e.,~
\begin{align} 
	\rho_0  := \frac1m \sum_{i,j=0}^{m-1} |ii\rangle\langle jj|
	\quad \text{for $m\leq d$}
\end{align}
against its orthogonal complement ${(\mathds{1}_{d^2}-\rho_0)}/{(d^2- 1)}$.
Equivalently, such state can be written as a maximally entangled state $\Phi_m$ on $\mathbb{C}^m \otimes \mathbb{C}^m$ embedded in a higher dimension Hilbert space $\mathbb{C}^d \otimes \mathbb{C}^d$.
In other words, we consider the binary hypotheses
\begin{align} \label{eq:MES_high}
	\begin{dcases}
		\mathsf{H}_0 : \rho_0^{\otimes n} = \left(\Phi_m \oplus \mathbb{O}_{d^2-m^2}\right)^{\otimes n} & \\
		\mathsf{H}_1 : \rho_1^{\otimes n} = \left( \frac{\mathds{1}_{d^2}-\Phi_m \oplus\mathbb{O}_{d^2-m^2} }{d^2-1}  \right)^{\otimes n} & \\
	\end{dcases}, \quad \forall n.
\end{align}

We first prove a useful converse bound below, which 
will be used for our analysis later.

\begin{lemm} \label{lemm:converse_high}
	Fix dimensions $d$ and $d'$.
	Let $\rho$ and $\sigma$ be density matrices on $\mathbb{C}^d \otimes \mathbb{C}^d$.
	Then, for any class of two-outcome POVM, $\textnormal{X}\in \{\textnormal{LOCC}, \textnormal{SEP}, \textnormal{PPT}, \textnormal{ALL} \}$, any $p, \lambda\in[0,1]$ and any natural number $n$, it follows that
	\begin{align}
		P_\textnormal{e}^\textnormal{X} \left( (\rho\oplus \mathbb{O}_{d'^2} )^{\otimes n}, \left( \lambda\sigma \oplus (1-\lambda) \tau_{d'^2} \right)^{\otimes n} ; p \right) \geq
		P_\textnormal{e}^\textnormal{X} \left( \rho^{\otimes n},   (\lambda \sigma)^{\otimes n} ; p \right).
	\end{align}
\end{lemm}

\begin{proof}
	Recalling the definition of optimal average error probability for any operators, (possibly sub-normalized), i.e.,~
	\begin{align}
		P_\text{e}(\rho_0, \rho_1; p ) := \inf_{  T   } \left\{  p \alpha (T) + (1-p) \beta (T) \right\},
	\end{align}
	it follows that
	\begin{align}
		&P_\textnormal{e}^\textnormal{X} \left( (\rho\oplus \mathbb{O}_{d'^2} )^{\otimes n}, \left( \lambda\rho^\perp \oplus (1-\lambda) \tau_{d'^2} \right)^{\otimes n} ; p \right) \notag \\&=
		\inf_{T_n \in \text{X}} \left\{ p \Tr\left[ (\mathds{1}-T_n) (\rho\oplus \mathbb{O}_{d'^2} )^{\otimes n} \right] + (1-p) \Tr\left[ T_n \left( \lambda\sigma \oplus (1-\lambda) \tau_{d'^2} \right)^{\otimes n} \right] \right\} \\
		&\geq \inf_{T_n \in \text{X}} \left\{ p \Tr\left[ (\mathds{1}-T_n) \rho^{\otimes n} \right] + (1-p) \Tr\left[ T_n \left( \lambda\sigma \right)^{\otimes n} \right] \right\} \notag \\
		&\quad +  \inf_{G_n \in \text{X} } \left\{ p \Tr\left[ (\mathds{1}-G_n) \mathbb{O}_{d'^2} \right] + (1-p) \Tr\left[ G_n \sum_{k=0}^{n-1} \lambda^k (1-\lambda)^{1-k} B_k \right] \right\} , \label{eq:converse_high1}
	\end{align}
	where we have used super-additivity of infimum, and $B_k$ denotes the sum of all elements of $\{ \sigma, \tau_{d'^2} \}^{\otimes n}$ which have $k$ copies of $\sigma$.
	Note that the second term in Eq.~\eqref{eq:converse_high1} is zero since one allows to choose the projection $\mathds{1}_{d'^2}$ onto the copy of $\tau_{d'^2}$ as the POVM. Moreover, this measurement is implementable by all four classes of POVMs. This then completes the proof.
\end{proof}

In the following proposition, we consider a more general case compared to 
Eq.~\eqref{eq:MES_high} since there are non-unique orthogonal complements in high-dimensional systems.

\begin{prop} \label{theo:MES_high}
	Consider the hypotheses
	\begin{align} \label{eq:MES_high_0}
		\begin{dcases}
			\mathsf{H}_0 : \rho_0^{\otimes n} = \left( \Phi_m  \oplus \mathbb{O}_{d^2-m^2}\right)^{\otimes n} & \\
			\mathsf{H}_1 : \rho_1^{\otimes n} = \left( \lambda \Phi_m^\perp \oplus   (1-\lambda) \tau_{d^2 - m^2}  \right)^{\otimes n} & \\
		\end{dcases}, \quad \forall n.
	\end{align}
	Then, it holds that
	\begin{align}
		P_\textnormal{e}^\textnormal{PPT} (\rho_0^{\otimes n}, \rho_1^{\otimes n}; p ) = P_\textnormal{e}^\textnormal{SEP} (\rho_0^{\otimes n}, \rho_1^{\otimes n}; p ) = P_\textnormal{e}^\textnormal{LOCC} (\rho_0^{\otimes n}, \rho_1^{\otimes n}; p ) = \min\left\{ (1-p) \left( \frac{\lambda}{m+1} \right)^n, p 
		\right\}.
	\end{align}
\end{prop}
\begin{proof} 
	We first prove the achievability, i.e.,~the ``$\geq$" direction.
	We choose the following POVM:
	\begin{align}
		&\quad\left\{
		\left( M_m \oplus \mathbb{O}_{d^2-m^2} \right)^{\otimes n},
		\mathbb{1}_{d^{2n}} - \left( M_m \oplus \mathbb{O}_{d^2-m^2} \right)^{\otimes n }
		\right\} \notag \\
		& = \left\{
		M_m^{\otimes n} \oplus \mathbb{O}_{d^{2n}-m^{2n}},
		\left( \mathds{1}_{m^{2n}} - M_m^{\otimes n} \right) \oplus \mathbb{1}_{d^{2n}-m^{2n}}
		\right\},
	\end{align}
	where $M_m = \Phi_m + \frac{1}{m+1} (\mathds{1}_{m^2} - \Phi_m) $ as in Eq.~\eqref{eq:Md}.
	
	Then, the average error probability for the chosen LOCC protocol is
	\begin{align}
		P_\text{e}^{\text{LOCC} } 
		&= p \Tr\left[ \left( \mathds{1}_{m^{2n}} - M_m^{\otimes n} \right) \oplus \mathbb{1}_{d^{2n}-m^{2n}} \cdot \Phi_m^{\otimes n}  \oplus \mathbb{O}_{d^{2n}-m^{2n}} \right]  \notag \\
		&\quad+ (1-p) \Tr\left[  M_m^{\otimes n} \oplus \mathbb{O}_{d^{2n}-m^{2n}} \cdot\left( \lambda \Phi_m^\perp \oplus   (1-\lambda) \tau_{d^2 - m^2}  \right)^{\otimes n} \right] \\
		&= (1-p) \Tr\left[  M_m^{\otimes n} 	\left( \lambda \Phi_m^\perp\right)^n \right] \\
		&= (1-p) 	\left(\frac{m^2-1}{d^2-1}\right)^n \left( \frac{\lambda}{m+1} \right)^n.
	\end{align}
	By applying a similar technique to the proof of Proposition~\ref{prop:achievability}, we choose $\rho_1^{\otimes n} = \left( \frac{\mathds{1}_{d^2} - \Phi_m \oplus \mathbb{O}_{d^2-m^2} }{d^2-1} \right)^{\otimes n}$ whenever $P_\text{e}^{\text{LOCC} }  >p$.
	Hence, we complete the proof of achievability.
	
	Next, we move on to prove the converse, i.e.,~the ``$\geq$" direction.
	Invoking Lemma~\ref{lemm:converse_high}, we have
	\begin{align}
		P_\textnormal{e}^\textnormal{PPT} (\rho_0^{\otimes n}, \rho_1^{\otimes n}; p ) \geq P_\textnormal{e}^\textnormal{PPT} (\Phi^{\otimes n}, (\lambda\Phi_m^\perp)^{\otimes n}; p ).
	\end{align}
	Following a similar argument in Proposition~\ref{prop:achievability}, we obtain the desired result.
\end{proof}

\begin{remark}
	By choosing $\lambda =  \frac{m^2-1}{d^2-1} $ in Eq.~\eqref{eq:MES_high_0}, it can be verified that the single copy in the alternative hypothesis $\mathsf{H}_1$ coincides the canonical orthogonal complement in the alternative hypothesis of Eq.~\eqref{eq:MES_high}, namely~$\frac{\mathds{1}_{d^2}-\Phi_m\oplus\mathbb{O}_{d^2-m^2}}{d^2-1}$.
	This answers the optimal average error probability of binary hypotheses Eq.~\eqref{eq:MES_high} considered at the beginning of this section, i.e.,~$\min\left\{ (1-p) \left( \frac{m-1}{d^2-1} \right)^n, p 
	\right\}$.
\end{remark}

The binary hypotheses considered in Eq.~\eqref{eq:MES_high_0} have a simple variant 
(by interchanging $\Phi_m$ and $\Phi_m^\perp$), for which 
we can immediately calculate its optimal average error.

\begin{prop} \label{theo:MES_high_1}
	Consider the binary hypotheses
	\begin{align} \label{eq:MES_high_1}
		\begin{dcases}
			\mathsf{H}_0 : \rho_0^{\otimes n} = \left( \Phi_m^\perp  \oplus \mathbb{O}_{d^2-m^2}\right)^{\otimes n}, & \\
			\mathsf{H}_1 : \rho_1^{\otimes n} = \left( \lambda \Phi_m \oplus   (1-\lambda) \tau_{d^2 - m^2}  \right)^{\otimes n} & \\
		\end{dcases}, \quad \forall n,
	\end{align}
	where $\lambda \in [0,1]$.
	Then, it holds that
	\begin{align}
		P_\textnormal{e}^\textnormal{PPT} (\rho_0^{\otimes n}, \rho_1^{\otimes n}; p ) = P_\textnormal{e}^\textnormal{SEP} (\rho_0^{\otimes n}, \rho_1^{\otimes n}; p ) = P_\textnormal{e}^\textnormal{LOCC} (\rho_0^{\otimes n}, \rho_1^{\otimes n}; p ) =  \min\left\{ p\left( \frac{1}{m+1} \right)^n, (1-p)\lambda^n \right\}.
	\end{align}
\end{prop}
\begin{proof} 
	We choose the POVM as
	\begin{align}
		\left\{	\left( \mathds{1}_{m^{2n}} - M_m^{\otimes n} \right) \oplus \mathbb{O}_{d^{2n}-m^{2n}}, M_m^{\otimes n} \oplus \mathbb{1}_{d^{2n}-m^{2n}} \label{eq:MES_high_13}
		\right\},
	\end{align}
	where ${M}_m = \Phi_m + \frac{1}{m+1} \left( \mathds{1}_{m^2} - \Phi_m \right)$, as defined in Eq.~\eqref{eq:Md}.
	
	The average error probability for the chosen LOCC protocol is
	\begin{align}
		P_\text{e}^{\text{LOCC} } 
		&= p \Tr\left[ M_m^{\otimes n} \oplus \mathbb{1}_{d^{2n}-m^{2n}}  \cdot (\Phi_m^\perp)^{\otimes n}  \oplus \mathbb{O}_{d^{2n}-m^{2n}} \right]  \notag \\
		&\quad+ (1-p) \Tr\left[ \left( \mathds{1}_{m^{2n}} - M_m^{\otimes n} \right) \oplus \mathbb{O}_{d^{2n}-m^{2n}} \cdot\left( \lambda \Phi_m \oplus  (1-\lambda) \tau_{d^2 - m^2}  \right)^{\otimes n} \right] \\
		&= p\Tr\left[ {M}_m^{\otimes n} 	 (\Phi_m^\perp)^{\otimes n} \right] + (1-p)\Tr\left[ \left( \mathds{1}_{m^{2n}} - \Phi_m^{\otimes n} \right)  \cdot\left( \lambda \Phi_m  \right)^{\otimes n} \right]\\
		&= p\Tr\left[ {M}_m^{\otimes n} 	 (\Phi_m^\perp)^{\otimes n} \right] + (1-p)\lambda^n \\
		&= p \left( \frac{1}{m+1} \right)^n. \label{eq:MES_high_11}
	\end{align}

	On the other hand, we can also choose the POVM as
	\begin{align}
		\left\{	 \mathds{1}_{m^{2n}} \oplus \mathbb{O}_{d^{2n}-m^{2n}}, \mathbb{O}_{m^{2n}} \oplus \mathbb{1}_{d^{2n}-m^{2n}}
		\right\}. \label{eq:MES_high_14}
	\end{align}
	Then, the corresponding average error probability is 
	\begin{align}
		P_\text{e}^{\text{LOCC} } 
		&= p \Tr\left[ \mathbb{O}_{m^{2n}}  \oplus \mathbb{1}_{d^{2n}-m^{2n}}  \cdot (\Phi_m^\perp)^{\otimes n}  \oplus \mathbb{O}_{d^{2n}-m^{2n}} \right]  \notag \\
		&\quad+ (1-p) \Tr\left[ \mathds{1}_{m^{2n}}  \oplus \mathbb{O}_{d^{2n}-m^{2n}} \cdot\left( \lambda \Phi_m \oplus  (1-\lambda) \tau_{d^2 - m^2}  \right)^{\otimes n} \right] \\
		&= (1-p) \Tr\left[ \mathds{1}_{m^{2n}}  \cdot\left( \lambda \Phi_m   \right)^{\otimes n} \right] \\
		&= (1-p)\lambda^n. \label{eq:MES_high_12}
	\end{align}
	Since both POVMs chosen in Eqs.~\eqref{eq:MES_high_13} and \eqref{eq:MES_high_14} are implementable by LOCC protocols, we minimize the average error probabilities given in Eqs.~\eqref{eq:MES_high_11} and \eqref{eq:MES_high_12} to arrive at the ``$\leq $'' direction of our result.
	
	For the other direction, i.e.,~``$\geq$'', we use Lemma~\ref{lemm:converse_high}, to obtain
	\begin{align}
		P_\textnormal{e}^\textnormal{PPT} (\rho_0^{\otimes n}, \rho_1^{\otimes n}; p ) \geq P_\textnormal{e}^\textnormal{PPT} ((\Phi_m^\perp)^{\otimes n}, (\lambda \Phi_m)^{\otimes n}; p ).
	\end{align}
	Following a similar arguments to Proposition~\ref{prop:achievability} obtains the desired result.
\end{proof}

In the following Propositions~\ref{theo:sym_high} and \ref{theo:sym_high_1}, we apply similar techniques as before to calculate the binary hypothesis of a symmetric state embedded in a high-dimensional system and its orthogonal complement.

\begin{prop} \label{theo:sym_high}
	Consider the binary hypothesis
	\begin{align} \label{eq:sym_high}
		\begin{dcases}
			\mathsf{H}_0 : \rho_0^{\otimes n} = \left(\sigma_m \oplus \mathbb{O}_{d^2-m^2}\right)^{\otimes n}, & \\
			\mathsf{H}_1 : \rho_1^{\otimes n} = \left( \lambda \sigma_m^\perp \oplus   (1-\lambda) \tau_{d^2 - m^2}  \right)^{\otimes n} & \\	
		\end{dcases}, \quad \forall n,
	\end{align}
	where $\lambda \in [0,1]$ is arbitrary.
	Then, it holds that
	\begin{align}
		P_\textnormal{e}^\textnormal{LOCC} (\rho_0^{\otimes n}, \rho_1^{\otimes n}; p )  = \min\left\{ p\left( \frac{m-1}{m+1} \right)^n, (1-p) \lambda^n
		\right\}.
	\end{align}
\end{prop}
\begin{proof}
	The proof follows similar reasoning as in Theorem~\ref{theo:MES_high}.
	
	We choose the POVM as
	\begin{align}
		\left\{	\left( \mathds{1}_{m^{2n}} - \bar{M}_m^{\otimes n} \right) \oplus \mathbb{O}_{d^{2n}-m^{2n}}, \bar{M}_m^{\otimes n} \oplus \mathbb{1}_{d^{2n}-m^{2n}}
		\right\},
	\end{align}
	where $\bar{M}_m := \frac{m-1}{m+1} \Pi_s + \Pi_a$ as in Eq. (13) of \cite{WM08}, and $\Pi_s$ and $\Pi_a$ are the projections onto the support of the symmetric and anti-symmetric subspaces, respectively.
	
	Then, the average error probability for the chosen LOCC protocol is
	\begin{align}
		P_\text{e}^{\text{LOCC} } 
		&= p \Tr\left[ \bar{M}_m^{\otimes n} \oplus \mathbb{O}_{d^{2n}-m^{2n}}  \cdot \sigma_m^{\otimes n}  \oplus \mathbb{1}_{d^{2n}-m^{2n}} \right]  \notag \\
		&\quad+ (1-p) \Tr\left[ \left( \mathds{1}_{m^{2n}} - M_m^{\otimes n} \right) \oplus \mathbb{O}_{d^{2n}-m^{2n}} \cdot	\left( \lambda \sigma_m^\perp \oplus   (1-\lambda) \tau_{d^2 - m^2}  \right)^{\otimes n}\right] \\
		&= p \Tr\left[  \bar{M}_m^{\otimes n} \sigma_m^{\otimes n} \right] \\
		&= p \left(\frac{m-1}{m+1} \right)^n
	\end{align}
	where we invoke Eq. (17) of \cite{WM08} in the last line. 
	On the other hand, we can also choose the POVM as shown in Eq.~\eqref{eq:MES_high_14} to obtain $P_\text{e}^{\text{LOCC} } = (1-p)\lambda^n$. We then choose the minimum of the two to complete the achievability.
	
	The converse follows from Lemma~\ref{lemm:converse_high} and \cite[Proposition 3]{WM08}.
\end{proof}

\begin{prop} \label{theo:sym_high_1}
	Consider the binary hypothesis
	\begin{align} \label{eq:sym_high_1}
		\begin{dcases}
			\mathsf{H}_0 : \rho_0^{\otimes n} = \left(\sigma_m^\perp \oplus \mathbb{O}_{d^2-m^2}\right)^{\otimes n}, & \\
			\mathsf{H}_1 : \rho_1^{\otimes n} = \left(  \lambda \sigma_m \oplus  (1-\lambda) \tau_{d^2 - m^2}  \right)^{\otimes n} & \\
		\end{dcases}, \quad \forall n,
	\end{align}
	for some $\lambda \in [0,1]$.
	Then, it holds that
	\begin{align}
		P_\textnormal{e}^\textnormal{LOCC} (\rho_0^{\otimes n}, \rho_1^{\otimes n}; p )  = \min\left\{ (1-p)\left( \lambda\frac{m-1}{m+1} \right)^n, p 
		\right\}.
	\end{align}
\end{prop}

\begin{proof}
	We choose the POVM as
	\begin{align}
		\left\{
		\bar{M}_m^{\otimes n} \oplus \mathbb{O}_{d^{2n}-m^{2n}},
		\left( \mathds{1}_{m^{2n}} - \bar{M}_m^{\otimes n} \right) \oplus \mathbb{1}_{d^{2n}-m^{2n}}
		\right\},
	\end{align}
	where $\bar{M}_m := \frac{m-1}{m+1} \Pi_s + \Pi_a$ as in equation (13) of \cite{WM08} and the proof of Theorem~\ref{theo:sym_high}.
	
	Then, the average error probability for the chosen LOCC protocol is
	\begin{align}
		P_\text{e}^{\text{LOCC} } 
		&= p \Tr\left[ \left( \mathds{1}_{m^{2n}} - \bar{M}_m^{\otimes n} \right) \oplus \mathbb{1}_{d^{2n}-m^{2n}} \cdot \sigma_m^{\otimes n}  \oplus \mathbb{O}_{d^{2n}-m^{2n}} \right]  \notag \\
		&\quad+ (1-p) \Tr\left[  \bar{M}_m^{\otimes n} \oplus \mathbb{O}_{d^{2n}-m^{2n}} \cdot	\left(  \lambda \sigma_m \oplus (1-\lambda) \tau{d^2 - m^2}  \right)^{\otimes n} \right] \\
		&= (1-p) \Tr\left[  \bar{M}_m^{\otimes n} 	\lambda^n \sigma_m^{\otimes n} \right] \\
		&= (1-p) \left( \lambda\frac{m-1}{m+1} \right)^n.
	\end{align}
	Finally, we choose $\rho_1^{\otimes n}$ whenever $P_\text{e}^{\text{LOCC} }  > p$.
	
	The converse follows Lemma~\ref{lemm:converse_high} and \cite[Proposition 3]{WM08}.
\end{proof}

\begin{remark}
	The Stein, Hoeffding, and strong converse exponents can be obtained by following the same arguments in Section~\ref{sec:MES_high}. Table~\ref{table:summary}  provides a summary of the results.
\end{remark}

\end{document}